%% file: arXiv.tex
\documentclass[11pt,letter]{article}
\usepackage{bbding}
\fussy
\usepackage[dvipsnames]{xcolor}
\usepackage{framed, xcolor}
\usepackage{tocvsec2}
\usepackage{datetime}
\usepackage{pifont}
\usepackage{subfigure}          % for subfigures in ACM
\usepackage{colortbl}
\usepackage{booktabs}
\usepackage{pdfsync}
\usepackage[pagebackref=false,citecolor=newblue,colorlinks=true,linkcolor=newblue,bookmarks=false]{hyperref}

\usepackage[font=scriptsize,bf]{caption}
\usepackage{tikz,subfigure}
\usepackage[active]{srcltx}
\definecolor{newblue}{rgb}{0.2,0.2,0.6} % Define the color BrickRed

\usepackage[margin=1in]{geometry}
\usepackage{algorithm}
\usepackage{algpseudocode}
\usepackage{datetime}
\usepackage{epsfig,amssymb,amsfonts,amsmath,amsthm}
\usepackage{multirow}
\usepackage[numbers,sort&compress,sectionbib]{natbib}

\newcommand{\rot}{\intercal}

\newcommand{\calL}{\mathcal{L}}

\newcommand{\eps}{\varepsilon}
\newcommand{\vol}{\mu}

\newcommand{\OPT}{\mathsf{OPT}}
\newcommand{\SOL}{\mathsf{SOL}}

\newcommand{\thmref}[1]{Theorem~\ref{thm:#1}}

\newcommand{\lemref}[1]{Lemma~\ref{lem:#1}}

\newcommand{\figref}[1]{Figure~\ref{fig:#1}}

\newcommand{\eq}[1]{\eqref{eq:#1}}

\renewcommand{\tilde}{\widetilde}
\renewcommand{\epsilon}{\varepsilon}

\newcommand{\mat}[1]{#1}

\newtheorem{thm}{Theorem}[section]  %[chapter]

\newtheorem{lem}[thm]{Lemma}

\newtheorem*{rem*}{Remark}

\newcommand{\ONE}{{$\mathsf{ALLONE}_s$}}
\newcommand{\DISJN}{{$\mathsf{DISJ}_{s,n}$}}
\newcommand{\DISJL}{{$\mathsf{DISJ}_{s,\ell}$}}

\title{\textbf{Communication-Optimal Distributed Clustering}\footnote{A preliminary version of this paper appears at the 30th
Annual Conference on Neural Information Processing Systems (NIPS), 2016.}}

\author{
  Jiecao Chen\footnote{Department of Computer Science, Indiana University, Bloomington,  USA.  Work supported in part by NSF
CCF-1525024 and IIS-1633215.
  Email: \texttt{jiecchen@indiana.edu}}
 \and
  He Sun\footnote{Department of Computer Science, University of Bristol, Bristol,  UK.
 \texttt{h.sun@bristol.ac.uk}} 
\and
   David P. Woodruff\footnote{IBM Research Almaden, San Jose,  USA. \texttt{dpwoodru@us.ibm.com}}
   \and
   Qin Zhang\footnote{Department of Computer Science, Indiana University, Bloomington, USA. Work supported in part by NSF
CCF-1525024 and IIS-1633215.
 Email: \texttt{qzhangcs@indiana.edu}}
}

\date{}

\begin{document}

\maketitle

\begin{abstract}
  %  Clustering large datasets is a fundamental problem with a number of applications in machine learning. Data is often collected on different sites and clustering needs to be performed in a distributed manner with low communication. We would like the quality of the clustering in the distributed setting to match that in the centralized setting for which all the data resides on a single site. In this work, we study both graph and geometric clustering problems in two distributed models: (1) a point-to-point model, and (2) a model with a broadcast channel. We give protocols in both models which we show are nearly optimal by proving almost matching communication lower bounds. Our work highlights the surprising power of a broadcast channel for clustering problems; roughly speaking, to  cluster $n$ points or $n$ vertices in a graph distributed across $s$ servers, for a worst-case partitioning the communication complexity in a point-to-point model is $n \cdot s$, while in the broadcast model it is $n + s$. Similar phenomenon holds for the geometric setting as well. We implement our algorithms and demonstrate this phenomenon on real life datasets, showing that our algorithms are also very efficient in practice.
Clustering large datasets is a fundamental problem with a number of applications in machine learning. Data is often collected on different sites and clustering needs to be performed in a distributed manner with low communication. We would like the quality of the clustering in the distributed setting to match that in the centralized setting for which all the data resides on a single site. In this work, we study both graph and geometric clustering problems in two distributed models: (1) a point-to-point model, and (2) a model with a broadcast channel. We give protocols in both models which we show are nearly optimal by proving almost matching communication lower bounds. Our work highlights the surprising power of a broadcast channel for clustering problems; roughly speaking, to spectrally cluster $n$ points or $n$ vertices in a graph distributed across $s$ servers, for a worst-case partitioning the communication complexity in a point-to-point model is $n \cdot s$, while in the broadcast model it is $n + s$. A similar phenomenon holds for the geometric setting as well. We implement our algorithms and demonstrate this phenomenon on real life datasets, showing that our algorithms are also very efficient in practice. 
\end{abstract}

\thispagestyle{empty}

\setcounter{page}{0}

\newpage

\tableofcontents

\thispagestyle{empty}

\setcounter{page}{0}

\newpage

\input{introduction}

\input{preliminaries}

\input{messagepassing}

\input{blackboard}

\input{geometric}

\input{experiment}

%{\bf Acknowledgement: }D.W. thanks support from the
%     XDATA program of the Defense Advanced
%     Research Projects Agency (DARPA),
%     Air Force Research Laboratory
%     contract FA8750-12-C-0323.

\bibliography{reference,qin}
\bibliographystyle{plain}

%\appendix

%\input{appendix}

\end{document}

%% file: introduction.tex
\section{Introduction}\label{sec:intro}

Clustering is a fundamental task in machine learning with widespread applications in data mining, computer vision, and social network analysis. Example applications of clustering include grouping similar webpages by search engines,  finding users with common interests in a social network, and identifying different objects in a picture or video. For these applications, one can model the objects that need to be clustered as points in Euclidean space $\mathbb{R}^d$, where the similarities of two objects are represented by the Euclidean distance between the two points. 
Then the task of clustering is to choose $k$ points as centers, so that the total distance between all input points to their corresponding closest center is minimized. Depending on different distance objective functions, three typical problems have been studied: $k$-means, $k$-median, and $k$-center.

The other popular approach for clustering is to model the input data as vertices of a graph, and the similarity between two objects is represented by the weight of the edge connecting the corresponding vertices. For this scenario, one is asked to partition the vertices into clusters so that the ``highly connected'' vertices belong to the same cluster. A widely-used approach for graph clustering  is \emph{spectral clustering}, which embeds the vertices of a graph into the points in $\mathbb{R}^k$ through the bottom $k$ eigenvectors of the graph's Laplacian matrix, and applies $k$-means on the embedded points.

%For instance, one typical scenario for clustering is to group the points in the geometric space. Here, one is given $n$ points in the Euclidean space, and asked to choose $k$ points so that the total distance between every input point to its closest center is minimized. Depending on different distance functions, one can define the $k$-means problem, $k$-median problem, and the $k$-center problem.

%Another typical example for clustering is to group the vertices in a graph, where the edge wight between any pair of vertices represents the similarity of two vertices, and one is asked to partition the vertices so that highly-connected vertices belong to the same cluster. A widely used algorithm for graph clustering in practice is the spectral clustering: the algorithm embeds vertices to points in a high-dimensional Euclidean space through the bottom eigenvectors of the Laplacian matrix associated with the graph, and applies $k$-means algorithm on the embedded points.

Both the spectral clustering and the geometric clustering algorithms mentioned above have been widely used in practice, and have been the subject of extensive theoretical and experimental studies over the decades. However, these algorithms are designed for the centralized setting, and are not applicable in the setting of large-scale datasets that are maintained remotely by different sites.
In particular, collecting the information from all the remote sites and performing a centralized clustering algorithm is infeasible due to high communication costs, and new distributed clustering algorithms with low communication cost need to be developed.

There are several natural communication models, and we focus on two of them: (1) a point-to-point model, and (2) a model with a broadcast channel. In the former, sometimes referred to as the {\it message-passing model}, there is a communication channel between each pair of users. This may be impractical, and the so-called {\it coordinator model} can often be used in place; in the coordinator model there is a centralized site called the coordinator, and all communication goes through the coordinator. This affects the total communication by a factor of two, since the coordinator can forward a message from one server to another and therefore simulate a point-to-point protocol. There is also an additional additive $O(\log s)$ bits per message, where $s$ is the number of sites, since a server must specify to the coordinator where to forward its message. In the model with a broadcast channel, sometimes referred to as the {\it blackboard model}, the coordinator has the power to send a single message which is received by all $s$ sites at once. This
can be viewed as a model for single-hop wireless networks.

In both models we study the total number of bits
communicated among all sites. Although the blackboard model is at least as powerful as the
message-passing model, it is often unclear how to exploit its power to obtain better bounds for specific problems.
Also, for a number of problems the communication complexity is the same in both models, such as computing
the sum of $s$ length-$n$ bit vectors modulo two, where each site holds one bit vector \cite{pvz16}, or estimating large moments \cite{wz12}. Still,
for other problems like set disjointness it can save a factor of $s$ in the communication
\cite{beopv13}.
%Understanding for which problems the broadcast channel can help is an important goal.  
%Both the spectral algorithms and algorithms for clustering points in the Euclidean space are centralized, and  require the global knowledge of the dataset. However, as most of the nowadays' big datasets are maintained in a distributed way by multiple sites, it is unrealistic for one site to gather the information of the whole datasets and perform a centralized algorithm. Hence, we need to develop distributed clustering algorithms, for which the overall communication cost among all sites needs to be considered. 

\subsection{Our contributions}  
%We study both graph clustering and geometric clustering problems in the two distributed models above. 
%
We present algorithms for graph clustering: for any $n$-vertex graph whose edges are arbitrarily partitioned across $s$ sites, our algorithms have communication cost $\tilde{O}(ns)$ in the message passing model, and have communication cost $\tilde{O}(n+s)$ in the blackboard model, where the $\tilde{O}$ notation suppresses polylogarithmic factors. The algorithm in the message passing model has each site send a {\it spectral sparsifier} of its local data to the coordinator, who then merges them in order to obtain a spectral sparsifier of the union of the datasets, which is sufficient for solving the graph clustering problem. Our algorithm in the blackboard model is technically more involved, as we show a particular recursive sampling procedure for building a spectral sparsifier can be efficiently implemented using a broadcast channel. It is unclear if other natural ways of building spectral sparsifiers can be implemented with low communication in the blackboard model. Our algorithms demonstrate the surprising power of the blackboard model for clustering problems. Since our algorithms compute spectral sparsifiers, they also have applications to solving symmetric diagonally dominant linear systems in a distributed model. Any such system can be converted into a system involving a Laplacian (see, e.g., \cite{ACKQWZ16}), from which a spectral sparsifier serves as a good preconditioner. 

Next we show that $\Omega(ns)$ bits of communication is necessary in the message passing model to even recover a constant fraction of a cluster, and $\Omega(n+s)$ bits of communication is necessary in the blackboard model. This shows the optimality of our algorithms up to poly-logarithmic factors. 

We then study clustering problems in constant-dimensional Euclidean space. We show for any $c>1$, computing a $c$-approximation for $k$-median, $k$-means, or $k$-center correctly with constant probability in the message passing model requires $\Omega(sk)$ bits of communication. We then strengthen this lower bound, and show even for {\it bicriteria} clustering algorithms, which may output a constant factor more clusters and a constant factor approximation, our $\Omega(sk)$ bit lower bound still holds. Our proofs are based on
communication and information complexity. 
Our results imply that existing algorithms~\cite{BEL13} for $k$-median and $k$-means with $\tilde{O}(sk)$ bits of communication, as well as the folklore parallel guessing algorithm for $k$-center with $\tilde{O}(sk)$ bits of communication, are optimal up to poly-logarithmic factors. For the blackboard model, we present an algorithm for $k$-median and $k$-means that achieves an $O(1)$-approximation using $\tilde{O}(s+k)$ bits of communication. This again separates the models.

We give empirical results which show that using spectral sparsifiers preserves the quality of spectral clustering surprisingly well in real-world datasets.
%In the message passing model and the blackboard model, we show that our algorithms constructing spectral sparsifiers save a huge amount of communication while achieving similar results compared to a centralized algorithm.
For example, when we partition a graph with over $70$ million edges (the {\tt Sculpture} dataset) into $30$ sites, only $6\%$ of the input edges are communicated in the blackboard model and $8\%$ are communicated in the message passing model, while the values of the normalized cut (the objective function of spectral clustering) given in those two models are at most $2\%$ larger than the ones given by the centralized algorithm, and the visualized results are almost identical. This is strong evidence that spectral sparsifiers can be a powerful tool in practical, distributed computation. When the number of sites is large, the blackboard model incurs significantly less communication than the message passing model, e.g., in the {\tt Twomoons} dataset when there are $90$ sites, the message passing model communicates $9$ times as many edges as communicated in the blackboard model, illustrating the strong separation between these models that our theory predicts.

%We give empirical results which show that using spectral sparsifiers preserves the quality of spectral clustering surprisingly well in real-world datasets. In the message passing model and the blackboard model, we show that our algorithms constructing spectral sparsifiers save a huge amount of communication while achieving similar results compared to a centralized algorithm. This is strong evidence that spectral sparsifiers can be a powerful tool in practical, distributed computation. When the number of sites is large, the blackboard model incurs significantly less communication than the message passing model, illustrating the strong separation between these models that our theory predicts. 
%Our experiments also show how algorithms in those models can be affected by different parameters.

\subsection{Related work}  There is a rich literature on spectral and geometric clustering algorithms from various aspects~(see, e.g.,~\cite{k-means++,nips02,PSZ15,luxburg07}). Balcan et al.~\cite{BEL13,BKLW14} and Feldman et al. \cite{FSS13} study distributed $k$-means (\cite{BEL13} also studies $k$-median), and present provable guarantees on the clustering quality. Very recently Guha et al. \cite{GLZ17} studied distributed $k$-median/center/means with outliers.  %Lammersen et al.~\cite{journals/mst/Lammersen0S15} study probabilistic geometric clustering problem in data streams.
Cohen et al.~\cite{journals/corr/CohenEMMP14} study dimensionality reduction techniques for the input data matrices that can be used for distributed $k$-means. The main takeaway  is that there is no previous work which develops protocols for spectral clustering in the common message passing and blackboard models, and lower bounds are lacking as well. For geometric clustering, while upper bounds exist (e.g.,~\cite{BEL13,BKLW14,FSS13}), no provable lower bounds in either model existed, and our main contribution is to show that previous algorithms are optimal. We also develop a new protocol in the blackboard model.  
%
%
%\textcolor{red}{Most of these previous studies only considered the clustering problem in the non-distributed setting, while a few reference~(e.g.,~\cite{BEL13}) presented distributed geometric clustering algorithms but did not optimize the communication complexity.}
%
%
%
%Czumaj and Sohler~\cite{journals/mst/CzumajS10} studied sublinear-time approximation algorithms for the geometric clustering problem in metric spaces. 

%\textcolor{red}{add some related works. also add all the work by Balcan, Sohler, Cohen, Musco, etc.}

%% file: preliminaries.tex
\section{Preliminaries\label{sec:preliminaries}}

Let 
 $G=(V,E,w)$ be an undirected   graph with $n$ vertices,   $m$ edges, and weight function $V\times V\rightarrow \mathbb{R}_{\geq 0}$. 
The set of neighbors of a vertex $v$ is represented by $N(v)$,
and its degree is $d_v=\sum_{u\sim v} w(u,v)$. The maximum degree of $G$ is defined to be $\Delta(G)=\max_{v}\{d_v\}$. 
For any set $S\subseteq V$, let $\mu(S)\triangleq\sum_{v\in S} d_v$.
For any sets $S, T\subseteq V$, we define
$w(S,T)\triangleq \sum_{u\in S, v\in T} w(u,v)$ to be the total weight of edges crossing $S$ and $T$.
We define the conductance of any set $S$ by
\[
\phi(S)=\frac{w(S, V\setminus S)}{\mu(S)}.
\]
For two sets $X$ and $Y$, the symmetric difference of $X$ and $Y$ is defined
as $X\triangle Y\triangleq (X\setminus Y)\cup (Y\setminus X)$. 
%For any two graphs $G_1=(V, E_1,w_1)$ and $G_2=(V, E_2,w_2)$ defined on the same vertex set, we write $G_1\boxplus G_2= (V,E,w)$ where $E=E_1\cup E_2$ and the weight function $w$ is defined by $w(u,v)=w_1(u,v)+w_2(u,v)$ for any $u,v\in V$.

For any matrix $A\in\mathbb{R}^{n\times n }$, let $\lambda_1(A)\leq\cdots \leq \lambda_n(A)=\lambda_{\max}(A)$ be the eigenvalues of $A$. For any two matrices $A, B\in\mathbb{R}^{n\times n}$, we write $A\preceq B$ to represent $B-A$ is positive semi-definite~(\textsf{PSD}). Notice that this condition implies that $x^{\rot}Ax\leq x^{\rot}Bx$ for any $x\in\mathbb{R}^n$. Sometimes we also use a weaker notation
$(1-\varepsilon)A\preceq_r B\preceq_r (1+\varepsilon)A$ to indicate that 
\[
(1-\varepsilon)x^{\rot}Ax\leq x^{\rot}Bx\leq (1+\varepsilon)x^{\rot}Ax
\] for all $x$ in the row span of $A$.

%\medskip

\subsection{Graph Laplacian}  The Laplacian matrix of $G$ is an $n\times n$ matrix $L_G$ defined by 
$L_G=D_G-A_G$, where $A_G$ is the adjacency matrix of $G$ defined by $A_G(u,v)=w(u,v)$, and $D_G$ is the $n\times n$ diagonal matrix with $D_G(v,v)=d_v$ for any $v\in V[G]$.  Alternatively, we can write $L_G$ with respect to a \emph{signed edge-vertex incidence matrix}:  we assign every edge $e=\{u,v\}$ an arbitrary orientation, and let $B_G(e,v)=1$ if $v$ is $e$'s head, $B_G(e,v)=-1$ if $v$ is $e$'s tail, and $B_G(e,v)=0$ otherwise.
We further define a diagonal matrix $W_G\in\mathbb{R}^{m\times m}$, where $W_G(e,e)=w_e$ for any edge $e\in E[G]$. Then, we can write $L_G$ as $L_G=B_G^{\rot}W_GB_G$. The \emph{normalized Laplacian matrix} of $G$ is defined by
$
\calL_G\triangleq \mat{D}_G^{-1/2}\mat{L}_G\mat{D}_G^{-1/2} = \mat{I}-\mat{D}_G^{-1/2}\mat{A}_G\mat{D}_G^{-1/2}$.
We sometimes drop the subscript $G$ when the underlying graph is clear from the context. 

%By definition, $B_G$'s each row, expressed by $b_i$, corresponds to an edge of $G$, and we define the \emph{leverage score} of $b_i$ by  $ \tau_i=b_i^{\rot}L^{+}_G b_i$, where $L^+_G$ is the \emph{pseudo-inverse} of the Laplacian matrix of $G$.

%\medskip

\subsection{Spectral sparsification}   For any  undirected and weighted graph $G=(V,E,w)$, we say a subgraph $H$ of $G$ with proper reweighting of the edges is a $(1+\varepsilon)$-spectral sparsifier if
\begin{equation}\label{eq:ssproperty}
(1-\varepsilon)L_G  \preceq  L_{H} \preceq  (1+\varepsilon)L_G.
\end{equation}
By definition, it is easy to show that, if we decompose the edge set of a graph $G=(V,E)$ into $E_1,\ldots,E_{\ell}$ for a constant
$\ell$ and $H_i$ is a  spectral sparsifier of $G_i=(V, E_i)$ for any $1\leq i\leq \ell$, then the graph formed by the union of  edge sets from $H_i$ is a spectral sparsifier of $G$.  It is known that, for any undirected graph $G$ of $n$ vertices, there is a $(1+\varepsilon)$-spectral sparsifier of $G$ with $O(n/\varepsilon^2)$ edges, and it can be constructed in almost-linear time~\cite{LS15}. 

The following lemma shows that a spectral sparsifier preserves the clustering structure of a graph.

\begin{lem}\label{lem:sssubset}
Let $H$ be a $(1+\varepsilon)$-spectral sparsifier of $G$ for some $\varepsilon\leq 1/3$. Then, it holds for any set $S\subseteq V$ that   $\phi_H(S)\in \left(\frac{1}{2},2\right)\phi_G(S)$.
\end{lem}

\begin{proof}
Let $x_u\in\mathbb{R}^n$ be the indicator vector of vertex $u$, i.e., $x_u(v)=1$ if $u=v$, and $x_u(v)=0$ otherwise. We have that \[
(1-\varepsilon)\cdot x_u^{\rot} L_Gx_u \leq x_u^{\rot} L_H x_u  \leq (1+\varepsilon)\cdot x_u^{\rot} L_G x_u,
\]
which  implies that  $
(1-\varepsilon)\cdot\mu_G(S)\leq \mu_H(S)\leq (1+\varepsilon)\cdot \mu_G(S)$ for any subset $S$.

Similarly, for any set $S\subseteq V$ we define the indicator vector of $S$ by $x_S\in\mathbb{R}^n$, where $x_S(u)=1$ if $u\in S$, and $x_S(u)=0$ otherwise. Hence, $x_S^{\rot}L_Gx_S=w_G(S,V\setminus S)$, and $x_S^{\rot}L_Hx_S=w_H(S,V\setminus S)$. Combining these with \eq{ssproperty}, we have that 
\[
(1-\varepsilon)\cdot w_G(S, V\setminus S)\leq w_H(S, V\setminus S)\leq (1+\varepsilon)\cdot w_G(S, V\setminus S).
\]
Hence, for any subset $S$ we have that 
\[
\phi_H(S) =\frac{w_H(S, V\setminus S)}{\vol_H(S)} \leq \frac{(1+\varepsilon) w_G(S, V\setminus S)}{(1-\varepsilon) \vol_G(S) } \leq 2\cdot\phi_G(S),
\]
where the last inequality holds by assuming $\varepsilon\leq 1/3$. Similarly, we have that
\[
\phi_H(S) =\frac{w_H(S, V\setminus S)}{\vol_H(S)} \geq \frac{(1-\varepsilon) w_G(S, V\setminus S)}{(1+\varepsilon) \vol_G(S) } \geq \frac{1}{2}\cdot\phi_G(S).
\]
Hence,   $\phi_H(S)$ and $\phi_G(S)$ differ by at most a factor of 2 for any vertex set $S$. 
\end{proof}

%\begin{lem}\label{lem:ssadd}
%Let $G_1$ and $G_2$ be two  graphs on the same vertex set, and $H_1$ and $H_2$ be  $(1+\varepsilon)$-spectral sparsifiers of $G_1$ and $G_2$ respectively, for some  $\varepsilon<1$. Then, $H_1\boxplus H_2$ is a $(1+\varepsilon)$-spectral sparsifier of $G_1\boxplus G_2$. 
%\end{lem} 
 
%\begin{proof}
%The statement holds by noticing that the Laplacian matrix of $G_1\boxplus G_2$~(as well as $H_1\boxplus H_1$) is the sum of Laplacian matrices of $G_1$ and $G_2$~(as well as $H_1$ and $H_2$).
%\end{proof}

\subsection{Models of computation}  We  study distributed clustering in two models for distributed data: the message passing model and the blackboard model. The message passing model represents those distributed computation systems with point-to-point communication, and the blackboard model represents those where messages can be broadcast to all parties. 

More precisely, in the message passing model there are $s$ sites $\mathcal{P}_1,\ldots, \mathcal{P}_s$, and one coordinator. These sites can talk to the coordinator through a two-way private channel. In fact, this is referred to as the coordinator model in Section \ref{sec:intro}, where it is shown to be equivalent to the point-to-point model up to small factors. The input is initially distributed at the $s$ sites. The computation is in terms of rounds: at the beginning of each round, the coordinator sends a message to some of the $s$ sites, and then each of those sites that have been contacted by the coordinator sends a message back to the coordinator.  At the end, the coordinator outputs the answer.   In the alternative blackboard model, the coordinator is simply a blackboard where these $s$ sites $\mathcal{P}_1,\ldots, \mathcal{P}_s$ can share information; in other words, if one site sends a message to the coordinator/blackboard then all the other $s-1$ sites can see this information without further communication. The order for the sites to speak is decided by the contents of the blackboard.

For both models we measure the \emph{communication cost} as the total number of bits sent through the channels.  %We comment that for most problems, including those in this paper, if the input is distributed into the $s$ sites such that each , then the communication cost between each site and the coordinator will also be balanced.  %
The two models are now standard in multiparty communication complexity (see, e.g., \cite{beopv13,pvz16,wz12}).  They are similar to the congested clique model~\cite{LPPP03}  studied in the distributed computing community; the main difference is that in our models we do not post any bandwidth limitations at each channel but instead consider the total number of bits communicated.

\subsection{Communication complexity}  For any problem $\mathcal{A}$ and any protocol $\Pi$ solving $\mathcal{A}$, the \emph{communication complexity} of a protocol $\Pi$ is the maximum communication cost of $\Pi$ over all possible inputs $X$.  When the protocol is randomised, we define 
the \emph{error} of $\Pi$   by 
\[
\max_X\mathbb{P}\left(\mbox{the coordinator outputs an incorrect answer on $X$}\right),
\]
where the $\max$ is over all inputs $X$ and the probability is over all random strings of the coordinator and $s$ sites.
 The \emph{$\delta$-error randomised communication complexity} $\mathsf{R}_{\delta}(\mathcal{A})$ of a problem $\mathcal{A}$ in the message passing model  is the minimum communication complexity of any randomised protocol $\Pi$ that solves $\mathcal{A}$ with error at most $\delta$.  
 
Let $\mu$ be an input distribution on $X$. We call a deterministic protocol $(\delta, \mu)\text{-error}$ if it gives the correct answer for $\mathcal{A}$ on at least a $1 - \delta$ fraction of all input pairs, weighted by the distribution $\mu$. We denote $\mathsf{D}_{\delta, \mu}(\mathcal{A})$ as the cost of the minimum-communication $(\delta, \mu)\text{-error}$ protocol. A standard lemma in communication complexity called Yao's minimax lemma shows that 
$\mathsf{R}_{\delta}(\mathcal{A}) \ge \max_\mu \mathsf{D}_{\delta, \mu}(\mathcal{A}).$
%Thus to prove a lower bound for randomized communication complexity we can just find a hard input distribution and prove deterministic communication complexity for that distribution.

\subsection{Information complexity}  We abuse notation by using $\Pi$ for both the protocol and its transcript (its concatenation of messages). In the message passing model, let $\Pi_i\ (i \in [s])$ be the transcript (set of messages exchanged) between the $i$-th site and the coordinator.  Then $\Pi$ can be seen as a concatenation $\Pi_1 \circ \Pi_2 \circ \ldots \circ \Pi_s$ ordered by the timestamps of the messages.
We define the information complexity of a problem $\mathcal{A}$ in the message passing model by 
\[
\mathsf{IC}_{\mu, \delta}(\mathcal{A}) = \min_{(\delta, \mu)\text{-error}\ \Pi} \sum_{i \in [s]} I(X_1, \ldots, X_s; \Pi_i),
\]where $I(\cdot\ ;\ \cdot)$ is the mutual information function. It has been shown in \cite{HRVZ15} that $\mathsf{R}_{\delta}(\mathcal{A}) \ge \mathsf{IC}_{\delta, \mu}(\mathcal{A})$ for any input distribution $\mu$.

%% file: messagepassing.tex
\section{Distributed graph  clustering\label{sec:graphmsg}}

In this section we study distributed graph clustering.   We assume that the vertex set of  the input graph $G=(V,E)$ can be partitioned into $k$ clusters, where vertices in each cluster $S$ are highly connected to each other, and there are fewer edges between $S$ and $V\setminus S$.  To formalize this notion, 
 we  define the \emph{$k$-way expansion constant} of graph $G$ by 
\[
\rho(k)\triangleq \min_{\mbox{\footnotesize partition $A_1,\ldots, A_k$}}\max_{1\leq i\leq k}\phi_G(A_i).
\] Notice that a graph $G$ has $k$ clusters if the value of $\rho(k)$ is small. It was shown in \cite{conf/stoc/LeeGT12} that $\rho(k)$ closely relates to  $\lambda_k(\mathcal{L}_G)$ by the following higher-order Cheeger inequality:
\begin{equation*}
\frac{\lambda_k(\mathcal{L}_G)}{2}\leq \rho(k)\leq O(k^2) \sqrt{\lambda_k(\mathcal{L}_G)}.
\end{equation*}
Hence, a large gap between $\lambda_{k+1}(\mathcal{L}_G)$ and $\rho(k)$ implies (i) the existence of 
a $k$-way partition $\{S_i\}_{i=1}^k$ such that every $S_i$ has small conductance
 $\phi_G(S_i)\leq \rho(k)$, and (ii) any $(k+1)$-way partition of $G$ contains a subset with  high conductance $\rho(k+1)\geq \lambda_{k+1}(\mathcal{L}_G)/2$. Therefore, a large gap between $\lambda_{k+1}(\mathcal{L}_G)$ and $\rho(k)$ ensures that $G$ has \emph{exactly} $k$ clusters. In the following, we assume that 
\[
\Upsilon\triangleq \lambda_{k+1}(\mathcal{L}_G)/\rho(k)=\Omega(k^3)
\]  
to ensure that the input graph $G$ has exactly $k$ clusters. The same assumption has been used 
 in the literature for studying graph clustering in the centralized setting~\cite{PSZ15}.

Both algorithms presented in the section are based on the following \emph{spectral clustering} algorithm: (i) compute the  $k$ eigenvectors $f_1,\ldots, f_k$ of $\mathcal{L}_G$ associated with $\lambda_1(\mathcal{L}_G),\ldots,\lambda_k(\mathcal{L}_G)$; (ii) embed every vertex $v$ to a point in $\mathbb{R}^k$ through the embedding 
\[
F(v)=\frac{1}{\sqrt{d_v}}\cdot(f_1(v),\ldots, f_k(v));
\]
 (iii) run $k$-means on the embedded points $\{F(v)\}_{v\in V}$, and group the vertices of $G$ into $k$ clusters according to the output of $k$-means.

\subsection{The message passing model\label{sec:gcmessage}}

We assume  the edges of the input graph $G=(V,E)$ are arbitrarily allocated among $s$ sites $\mathcal{P}_1,\cdots, \mathcal{P}_s$, and we use  $E_i$ to denote the edge set maintained by site $\mathcal{P}_i$. Our proposed algorithm consists of two steps: (i) every  $\mathcal{P}_i$ computes a linear-sized $(1+\eps)$-spectral sparsifier $H_i$ of $G_i\triangleq (V, E_i)$, 
for a small constant $\eps\leq 1/10$,
and sends the edge set of  $H_i$, denoted by $E_i'$, to the coordinator; (ii) the coordinator runs a spectral clustering algorithm on the union of received graphs $H\triangleq \left(V,\bigcup_{i=1}^k E_i' \right)$. The theorem  below summarizes the performance of this algorithm, and shows the approximation guarantee of this algorithm is as good as the provable guarantee of spectral clustering known in the centralized setting, which is shown in the lemma below.

\begin{lem}[\cite{PSZ15}]\label{lem:PSZ}
Let $G$ be a graph satisfying the condition $\Upsilon=\Omega(k^3)$, and $k\in \mathbb{N}$.
Then, a spectral clustering algorithm outputs sets $A_1,\ldots, A_k$ such that $\vol(A_i\triangle S_i)= O\left(k^3\cdot\Upsilon^{-1}\cdot\vol(S_i)\right)$ holds for any $1\leq i\leq k$, where $S_i$ is the optimal cluster corresponding to $A_i$.
\end{lem}

\begin{thm}\label{thm:gcss}
Let $G=(V,E)$ be an $n$-vertex graph with $\Upsilon=\Omega(k^3)$, and suppose the edges of $G$ are
arbitrarily allocated among $s$ sites.  Assume  $S_1,\cdots, S_k$ is an optimal partition that achieves $\rho(k)$.
Then, the algorithm above computes  a partition $A_1,\ldots, A_k$ satisfying $\vol(A_i\triangle S_i)= O\left( k^3\cdot\Upsilon^{-1}\cdot\vol(S_i)\right)$  for any $1\leq i\leq k$. The total communication cost of this algorithm is $\tilde{O}(ns)$ bits. 
\end{thm}

\begin{proof}
By the definition of the Laplacian matrix, we have that $L_G=\sum_{i=1}^s L_{G_i}$. Since every $H_i$ is a $(1+\eps)$-spectral sparsifier of graph $G_i$,   we have that 
$(1-\eps) L_{H_i}\preceq L_{G_i} \preceq (1+\eps) L_{H_i}$.
This implies that $
(1-\eps) L_{H}\preceq L_{G} \preceq (1+\eps) L_{H}$,
by the definition of $H_i$ and graph Laplacians. 
Now we show that our assumption on $\Upsilon$ is preserved in $H$. By \lemref{sssubset}, we have for any $1\leq i \leq k$ that 
$\phi_H(S_i)\in \left(\frac{1}{2},2\right) \phi_G(S_i)$,
which implies that $S_i$ has low conductance in $H$, and $\rho_H(k)\in \left(\frac{1}{2},2\right)\rho_G(k)$. To show that $\lambda_k(\mathcal{L}_H)$ is a constant approximation of $\lambda_k(\mathcal{L}_G)$, notice that \[(1-c)\cdot x^{\rot} L_Gx \leq x^{\rot} L_H x \leq (1+c)\cdot x^{\rot} L_G x\] holds for any $x\in\mathbb{R}^n$. Hence
it holds for any $x\in\mathbb{R}^n$ that 
\[
(1-\varepsilon)\cdot x^{\rot}D_G^{-1/2} L_GD_G^{-1/2}x \leq x^{\rot} D_G^{-1/2}L_H D_G^{-1/2}x \leq (1+\varepsilon)\cdot x^{\rot}D_G^{-1/2} L_GD_G^{-1/2} x.\]
Since $D_G^{-1/2} L_GD_G^{-1/2}=\mathcal{L}_G$ and  $\frac{1}{2} D_G^{-1}\preceq D_H^{-1}\preceq 2 D_G^{-1}$, we have that $\lambda_i\left(\mathcal{L}_H\right)=\Theta\left(\lambda_i\left(\mathcal{L}_G\right)\right)$, and the  assumption on $\Upsilon$ in $H$ is preserved from $G$ up to a constant factor.  By \lemref{PSZ}, the output of a spectral clustering algorithm on $H$ satisfies the claimed properties.
The total communication cost of $\tilde{O}(ns)$  bits follows from the fact that every $H_i$ has  $O(n)$ edges.
\end{proof}

Next we show that the communication cost of our proposed algorithm is optimal up to a logarithmic factor. 
 Our analysis is based on a reduction from graph clustering to the Multiparty Set-Disjointness problem~(\DISJN): for any $s$ sites  $\mathcal{P}_1,\ldots, \mathcal{P}_s$, where each  $\mathcal{P}_i$ has a set $S_i\subseteq [n]$, let $X_i = (X_i^1, \ldots, X_i^n)$ be the characteristic vector of $S_i$, and let $X = (X_1, \ldots, X_s)$ be the input matrix with $X_i$ being the $i$-th row.  Let $X^j = (X_1^j, \ldots, X_s^j)$ be the $j$-th column of the input matrix $X$. 
We define a function \ONE\ on an $s$-bit vector $Y = (Y_1, \ldots, Y_s)$ as
$
\text{\ONE} (Y)  =  \bigwedge_{i \in [s]} Y_i$, and $ \text{\DISJN}(X) = \bigvee_{j \in [n]} \text{\ONE} (X^j)$. Then the \DISJN\ problem asks the value of $\text{\DISJN}(X)$.
We  introduce two hard input distributions for \ONE\ and \DISJN\ respectively.  
\begin{itemize}
\item {\em Hard input distribution $\nu$ on $Y \in \{0,1\}^s$ for \ONE:} with probability $1/2$, we choose each $Y_i\ (i \in [s])$ to be $0$ or $1$ with equal probability; with probability $1/4$ we choose $Y$ to be an all-$1$ vector; and with the remaining probability $1/4$ we choose $Y$ to be a random vector with $n-1$ coordinates being $1$'s and a random coordinate being $0$.
\item {\em Hard input distribution $\mu_n$ on $X \in \{0,1\}^{s \times n}$ for \DISJN:}  For each $j \in [n]$, we choose $X^j \sim \nu$.  
\end{itemize}
%We will use the following results by Braverman et al.

\begin{thm}[\cite{beopv13}] \label{thm:DISJ} 
It holds that $\mathsf{IC}_{0.49, \nu}(\text{\ONE}) = \Omega(s)$, and 
  $\mathsf{IC}_{0.49, \nu}(\text{\DISJN}) = \Omega(sn)$.
\end{thm}

\begin{lem} \label{lem:DISJ}
In the message passing model, any randomized algorithm that computes \DISJN\ correctly with probability $0.9$ needs $\Omega(sn)$ bits of communication.
\end{lem}

\begin{proof}
The lemma follows from \thmref{DISJ} and Yao's minimax lemma.
\end{proof}

\begin{thm}\label{thm:lbmsg}
Let $G$ be an undirected graph with $n$ vertices, and suppose the edges of $G$ are distributed among $s$ sites. Then, any algorithm that correctly outputs a constant fraction of a cluster in $G$ requires $\Omega(ns)$ bits of communication. This holds even if each cluster has constant expansion.  
\end{thm}

\begin{proof}
Our proof is based on the reduction from graph clustering to the Multiparty Set-Disjointness problem~(\DISJN).
For any item $j$ and site $\mathcal{P}_i$, we set $X^j_i=0$ if item $j$ appears in site $\mathcal{P}_i$, and $X^j_i=1$ otherwise.
Then $\text{\DISJN}(X)=1$ if there is some item not appearing in any site.  Now we construct a graph $G$ based on the hard instance $X$ of 
$\mathsf{DISJ}_{n,s}$ as follows: initially,  graph $G$ consists of  $n$ isolated vertices $\ell_1,\ldots, \ell_n$, and $r$ isolated vertices $r_1,\ldots, r_s$. Then, we add an edge between $\ell_j$ and $r_i$ if item $j$ appears in site $\mathcal{P}_i$. With this construction, it is easy to see that $\text{\DISJN}(X)=0$ if every vertex $\ell_j$ is connected to some $r_i$, and $\text{\DISJN}(X)=1$ if there are some isolated vertices $\ell_j$.

We will show that, when $\text{\DISJN}(X)=0$, our constructed  graph $G$ is a bipartite expander, i.e., $G$ has only $1$ cluster.   To prove this, notice that, from the hard input distribution $\mu$ on $Y\in\{0,1\}^s$  described above, with probability $1/2$ we choose each $Y_i(i\in [s])$ to be $0$ or $1$ with equal probability. This implies that, for any $\ell_i$ and $r_j$, there is an edge between $\ell_i$ and $r_j$ independently with probability at least $1/4$. By standard results on constructing expanders, this implies $G$ is a bipartite expander with constant expansion, and in particular is connected.   

On the other side, when $\text{\DISJN}(X)=1$, every isolated vertex $\ell_j$ itself forms a cluster with conductance $0$ and constant expansion, and the giant component of $G$ forms a cluster with conductance $0$ and constant expansion (since, as argued in the previous paragraph, it is a bipartite expander). Let $k$ be the number of connected  components in graph $G$. Then, $\rho(k)=0$, and our assumption on $\Upsilon=\lambda_{k+1}(\mathcal{L}_G)/\rho(k)=\Omega(k^3)$ holds trivially. Hence any clustering algorithm  that is able to find a constant fraction of each cluster in graph $G$ satisfying $\Upsilon=\Omega(k^3)$ can be used to solve \DISJN. The lower bound on the communication complexity of graph clustering follows from the lower bound for \DISJN.  
\end{proof}

As a  remark, it is easy to see that this  lower bound also holds for constructing spectral sparsifiers: for any $n\times n$ $\mathsf{PSD}$ matrix $A$ whose entries are arbitrarily distributed among $s$ sites, any distributed algorithm that constructs a $(1+\Theta(1))$-spectral sparsifier of $A$ requires $\Omega(ns)$ bits of communication. This follows since such a spectral sparsifier can be used to solve the spectral clustering problem. Spectral sparsification has played an important role in designing fast algorithms from different areas, e.g., machine learning, and  numerical linear algebra. Hence our lower bound result for constructing spectral sparsifiers may have applications to studying other distributed learning algorithms.

%% file: blackboard.tex
\subsection{The blackboard model\label{sec:bb}}
Next we present a graph clustering algorithm with $\tilde{O}(n+s)$ bits of communication cost in the blackboard model. 
Our result is based on the observation that a spectral sparsifier preserves the structure of clusters, which was used for proving \thmref{gcss}. So it suffices to design a  distributed algorithm for constructing a spectral sparsifier in the blackboard model.

Our distributed algorithm  is based on constructing a chain of coarse sparsifiers~\cite{MP12}, which is described as follows: for any input \textsf{PSD} matrix $K$ with $\lambda_{\max}(K)\leq \lambda_u$ and all the non-zero eigenvalues of $K$  at least $\lambda_{\ell}$,   we define $d=\lceil \log_2(\lambda_u/\lambda_{\ell})\rceil$ and construct a chain of $d+1$ matrices
\begin{equation}\label{eq:chain}
[K(0),K(1),\ldots, K(d)],
\end{equation}
where $ 
\gamma(i)=\lambda_u/2^i$ and $K(i)= K+\gamma(i)I$.
Notice that in the chain above every $K(i-1)$
 is obtained by adding weights to the diagonal entries of $K(i)$, and 
 $K(i-1)$ approximates $K(i)$ as long as the weights added to the diagonal entries are  small.  
We will construct this chain recursively, so that $K(0)$ has heavy diagonal entries and can be approximated by a diagonal matrix. Moreover, since $K$ is the Laplacian matrix of a graph $G$, it is easy to see that $d=O(\log n)$ as long as the edge weights of $G$  are polynomially upper-bounded in $n$.

\begin{lem}[\cite{MP12}] \label{lem:MP12}
The chain~\eq{chain} satisfies the following relations: (1) 
 $K\preceq_r K(d)\preceq_r 2K$;
(2)  $K(\ell) \preceq K(\ell-1) \preceq 2K(\ell)$ for all $\ell\in\{1,\ldots, d\}$;
(3)  $K(0)\preceq 2\gamma(0)I\preceq 2K(0)$.
\end{lem}

%Such a chain can be constructed recursively so that $K(0)$ has heavy diagonal entries, and can be approximated by a diagonal matrix. Moreover, it is easy to prove that $d=O(\log n)$ if $K$ is the Laplacian matrix of a graph $G$, provided $G$'s edge weights are polynomially bounded.

Based on \lemref{MP12}, we will construct a chain of matrices  
\begin{equation}\label{eq:aptchain}
\left[ \tilde{K}(0), \tilde{K}(1),\ldots, \tilde{K}(d) \right]
\end{equation}
in the blackboard model, 
such that every $\tilde{K}(\ell)$ is a spectral sparsifier of $K(\ell)$, and 
every $\tilde{K}(\ell+1)$ can be constructed from $\tilde{K}(\ell)$.  The basic idea behind our construction is to use the relations among different $K(\ell)$ shown in \lemref{MP12} and the fact that, for any $K=B^{\rot}B$, sampling rows of $B$ with respect to their leverage scores can be used to obtain a matrix approximating $K$.

%It is well known that, for  $K(\ell)\triangleq B^{\rot}B$,  sampling rows $b_i$ of $B$ according to their leverage scores, denoted by$\tau_i\triangleq b_i^{\rot}K^{+}b_i$, will give a matrix approximating $K$. Formally, we assume that $\tilde{\tau}$ is the vector of leverage score overestimate for $B$'s rows such that $\tilde{\tau}_i\geq \tau_i$ for all $i\in[m]$, $0<\varepsilon<1$, and $c$ is a fixed constant. We sample every row $b_i$ with  probability $p_i=\min\left\{1, c\varepsilon^{-2}\tilde{\tau}_i\log n\right\}$, and define a diagonal matrix $W$ with $W_{i,i}=\frac{1}{p_i}$ with probability $p_i$, and $W(i,i)=0$ otherwise. Then, it holds with high probability that $(1-\varepsilon)K(\ell)\preceq \tilde{K}(\ell)=B^{\rot}WB \preceq (1+\varepsilon)K(\ell)$.Moreover, $W$ has $O\left( \|\tilde{\tau} \|_1\varepsilon^{-2}\log n \right)$ non-zeros with high probability. We prove in  Appendix~\ref{sec:b3} that the above sampling procedure can be implemented in the blackboard model, and this gives the following result:

\begin{thm}\label{thm:gcblackboard}
Let $G$ be an undirected graph on $n$ vertices, where the edges of $G$ are allocated among $s$ sites, and the edge weights are polynomially upper bounded in $n$. Then, a spectral sparsifier of $G$ can be constructed with $\tilde{O}(n+s)$ bits of communication in the blackboard model. That is, the chain \eq{aptchain} can be constructed with $\tilde{O}(n+s)$ bits of communication in the blackboard model.
\end{thm}

\begin{proof} Let $K= B^{\rot}B$  be the Laplacian matrix of the underlying graph $G$, where $B\in\mathbb{R}^{m\times n}$ is the  edge-vertex incidence matrix of $G$. We will  prove  that 
 every $\tilde{K}(i+1)$ can be constructed based on $\tilde{K}(i)$ with  $\tilde{O}(n+s)$ bits of communication. This implies that $\tilde{K}(d)$, a $(1+\varepsilon)$-spectral sparsifier of $K$, can be constructed with $\tilde{O}(n+s)$ bits of communication, as the length of the chain $d=O(\log n)$.

 First of all, notice that  $\lambda_u\leq 2n$, and the value of $n$ can be obtained with communication cost $\tilde{O}(n+s)$~(different sites sequentially write the new IDs of the vertices on the blackboard). In the following  we assume that $\lambda_u$ is the upper bound of $\lambda_{\max}$ that we actually obtained in the blackboard.
 
\emph{Base case of $\ell=0$:} 
 By definition, $K(0)=K+\lambda_u\cdot  I$, and 
$
 \frac{1}{2}\cdot K(0)\preceq \gamma(0)\cdot I \preceq K(0)$, 
due to Statement~3 of  \lemref{MP12}.  Let $\oplus$ denote appending the rows of one matrix to another. We define $B_{\gamma(0)}=B\oplus \sqrt{\gamma(0)}\cdot I$, and write $K(0)=K+\gamma(0)\cdot I=B^{\rot}_{\gamma(0)}B_{\gamma(0)}$. By defining $\tau_i=b_i^{\rot}\left(K(0)\right)^{\rot}b_i$ for each row of $B_{\gamma(0)}$, we have
$
\tau_i\leq b_i^{\rot}\left(\gamma(0)\cdot I\right) b_i\leq 2\cdot \tau_i$.
Let $\tilde{\tau}_i=b_i^{\rot}\left( \gamma(0)\cdot I \right)^{+}b_i$ be the leverage score of $b_i$ approximated using $\gamma(0)\cdot I$, and let $\tilde{\tau}$ be the vector of approximate leverage scores, with the leverage scores of the $n$ rows corresponding to $\sqrt{\gamma(0)}\cdot I$ rounded up to 1. Then, with high probability sampling $O(\varepsilon^{-2}n\log n)$ rows of $B$ will give a matrix $\tilde{K}(0)$ such that $
(1-\varepsilon)K(0)\preceq \tilde{K}(0)\preceq (1+\varepsilon)K(0)$.
Notice that, as every row of $B$ corresponds to an edge of $G$, the approximate leverage scores $\tilde{\tau}_i$ for different edges can be computed locally by different sites maintaining the edges, and the sites only need to send the information of the sampled edges to the blackboard, hence the communication cost is $\tilde{O}(n+s)$ bits.

\emph{Induction step:} We assume that 
$
(1-\varepsilon)K(\ell)\preceq_r\tilde{K}(\ell) \preceq_r (1+\varepsilon)K(\ell)$,
and the blackboard maintains the matrix $\tilde{K}(\ell)$. This implies that
$
(1-\varepsilon)/(1+\varepsilon)\cdot K(\ell)\preceq_r 1/(1+\varepsilon)\cdot \tilde{K}(\ell) \preceq_r K(\ell).
$
Combining this with Statement~2 of \lemref{MP12}, we have that 
\[
\frac{1-\varepsilon}{2(1+\varepsilon)}K(\ell+1)\preceq_r \frac{1}{2(1+\varepsilon)}\tilde{K}(\ell) \preceq K(\ell+1).
\]
We apply the same sampling procedure as in the base case, and obtain a matrix $\tilde{K}(\ell+1)$ such that $
(1-\varepsilon) K(\ell+1) \preceq_r\tilde{K}(\ell+1) \preceq_r (1+\varepsilon) K(\ell+1)$.
Notice that, since $\tilde{K}(\ell)$ is written on  the blackboard, the probabilities used for sampling individual edges can be computed locally by different sites, and in each round only the sampled edges will be sent to the blackboard in order for the blackboard to obtain $\tilde{K}(\ell+1)$. Hence, the total communication cost in each iteration is $\tilde{O}(n+s)$ bits. Combining this with the fact that the chain length $d=O(\log n)$ proves the theorem. 
\end{proof}

Combining \thmref{gcblackboard} and the fact that a spectral sparsifier preserves the structure of clusters, 
we obtain a distributed  algorithm in the blackboard model with total communication cost $\tilde{O}(n+s)$ bits, and the performance of our algorithm is the same as in the statement of \thmref{gcss}. Notice that $\Omega(n+s)$ bits of communication are needed for graph clustering in the blackboard model, since the output of a clustering algorithm contains $\Omega(n)$ bits of information and each site needs to communicate at least one bit. Hence the communication cost of our proposed algorithm is optimal up to a poly-logarithmic factor.

%% file: geometric.tex
\section{Distributed geometric clustering\label{sec:geo}}
We now consider geometric clustering, including $k$-median, $k$-means and $k$-center.  Let $P$ be a set of points of size $n$ in a metric space with distance function $d(\cdot, \cdot)$, and let $k \le n$ be an integer. In the $k$-center problem we want to find a set $C\ (|C| = k)$ such that $\max_{p \in P} d(p, C)$ is minimized, where $d(p, C) = \min_{c \in C} d(p, c)$.  In $k$-median and $k$-means we replace the objective function $\max_{p \in P} d(p, C)$ with $\sum_{p \in P} d(p, C)$ and $\sum_{p \in P} (d(p, C))^2$, respectively.
%We will assume that the input consists of points in the Euclidean space.

\subsection{The message passing model}
\label{sec:geo-mp}

As mentioned, for constant dimensional Euclidean space and a constant $c > 1$, there are algorithms that $c$-approximate $k$-median and $k$-means using $\tilde{O}(sk)$ bits of communication~\cite{BEL13}.  For $k$-center, the folklore parallel guessing algorithms (see, e.g., \cite{CMZ07}) achieve a $2.01$-approximation using $\tilde{O}(sk)$ bits of communication.  

The following theorem states that the above upper bounds are tight up to logarithmic factors. % Due to space constraints we defer the proof to the full version of this paper. %Appendix~\ref{app:proof-geometric}.  
The proof uses tools from multiparty communication complexity.  We in fact can prove a stronger statement that any algorithm that can differentiate whether we have $k$ points or $k+1$ points in total in the message passing model needs $\Omega(sk)$ bits of communication.

\begin{thm}
\label{thm:geometric}
For any $c > 1$, computing $c$-approximation for $k$-median, $k$-means or $k$-center correctly with probability $0.99$ in the message passing model needs $\Omega(sk)$ bits of communication.
\end{thm}

\begin{proof}
We can in fact prove a more general result: we can show that the $\Omega(sk)$ lower bound holds for any {\em eligible} function which evaluates to $0$ if there are at most $k$ points, and evaluates to a value greater than $0$ if there are at least $k+1$ points.  Note that $k$-median, $k$-means and $k$-center are all eligible functions.  We prove this by a simple reduction from \DISJL\ where $\ell = (k+1)/2$ (w.l.o.g., assuming $k$ is odd). 

The reduction is as follows. Given an $s$-player set-disjointness instance of size $\ell$ (i.e., \DISJL), let $X_i = (X_i^1, \ldots, X_i^\ell)$ be the $i$-th row of the input matrix $X$.  Let $p^1, \ldots, p^\ell$ and $q^1, \ldots, q^\ell$ be $2\ell$ distinct point locations on a line under Euclidean distance.  Each site $i$ does the following: for each coordinate $j$, if $X_i^j = 0$ then it put a point $u_i^j$ at location $q^j$; otherwise if $X_i^j = 1$ it put a point at location $p^j$.  It is easy to see that $\mathsf{DISJ}_{s,\ell} = 1$ if and only if the number of distinct points in $\bigcup_{i \in [s], j \in [\ell]} u_i^j$ is $2(\ell - 1) + 1 = k$; and $\mathsf{DISJ}_{s,\ell} = 1$ if and only if the number of distinct points in $\bigcup_{i \in [s], j \in [\ell]} u_i^j$ is $2(\ell - 1) + 2 = k  + 1$.  The lower bound follows from the definition of eligible function and Theorem~\ref{thm:DISJ}.
\end{proof}

A number of works on clustering consider {\em bicriteria} solutions (e.g., \cite{KPR98,CKMN01}). An algorithm is a $(c_1, c_2)$-approximation $(c_1, c_2 > 1)$ if the optimal solution costs $W$ when using $k$ centers, then the output of the algorithm costs at most $c_1 W$ when using at most $c_2 k$ centers.  We can show that for $k$-median and $k$-means, the $\Omega(sk)$ lower bound  holds even for algorithms with bicriteria approximations.  
\begin{thm}
\label{thm:bicriteria}

For any $c \in [1, 1.01]$, computing $(7.1 - 6c, c)$-bicriteria-approximation for $k$-median or $k$-means  correctly with probability $0.99$ in the message passing model needs $\Omega(sk)$ bits of communication.
\end{thm}

%\subsection{Proof of Theorem~\ref{thm:bicriteria}}
\label{app:proof-bicriteria}
%\he{To Qin: I copied everything you wrote from the appendix of our NIPS submission here. It seems unusual to put the lemmas/theorems inside a proof. Can you make some changes here?}
Before proving Theorem \ref{thm:bicriteria}, we first show the following technical lemma.  
\begin{lem}
\label{lem:direct-sum}
In the message-passing model, $\Omega(s \ell)$ bits of communication is needed for computing at least a $0.8$ fraction of $j \in [\ell]$ \ONE$(X^j)$ correctly with probability $0.99$ under the input distribution $X \sim \mu_\ell$.
\end{lem}

\begin{proof}
By a Markov inequality, there must exist $\Omega(s)$ coordinates $j$ such that the algorithm computes \ONE$(X^j)\ (X^j \sim \nu)$  with error probability at most $0.24$. Call each of these coordinates $j$ {\em good}. Let $\Pi$ be the protocol transcript. We have
\begin{eqnarray*}
I(X; \Pi) &=& \sum_{j \in [\ell]} I(X^j; \Pi\ |\ X^{-j}) \\
&\ge& \sum_{j \in [\ell]} I(X^j; \Pi) \quad (X^j \text{ and } X^{-j} \text{ are independent})\\
&\ge& \sum_{\text{good } j} I(X^j; \Pi)  \\
&\ge& \Omega(s) \cdot \mathsf{IC}_{0.24, \nu}(\text{\ONE})\\
&\ge& \Omega(s \ell). \quad \quad (\text{Theorem}\ \ref{thm:DISJ})
\end{eqnarray*}
\end{proof}

%We again prove by a reduction from the $s$-player set-disjointness, but the arguments are more involved.   Our proof holds for both $k$-median and $k$-means.

Now we are ready to prove the theorem.
%{\bf The reduction.}  
\begin{proof}[Proof of  \thmref{bicriteria}]
We consider $8 \ell$ point locations on a line (under Euclidean distance) with $x$-coordinates being $1, 2, \ldots, 8 \ell$. We put a point with {\em infinite} weight at every even point location. We name the $4 \ell$ odd point locations from left to right as 
$$p^1, q^1, p^2, q^2, \ldots, p^{\ell}, q^{\ell}, z^1, z^2, \ldots, z^{2\ell}.$$  
For each site $i \in [s]$ and each column $j \in [\ell]$, if $X_i^j = 0$ then we put a point with weight $1$ at location $q^j$; otherwise if $X_i^j = 1$ then we put a point with weight $1$ at location $p^j$.  We also put a point with weight $1/2$ at each of the ``dummy'' locations $z^1, \ldots, z^{2\ell}$.  Let the weight of a {\em location} be the sum of the weights of points falling into that location.

Given such an input $X$, for both $k$-median and $k$-means, the optimal solution ($\OPT$) which is allowed to use $k = 6\ell$ centers will include all locations $p^j$ and $q^j$ whose weights are at least $1$ (note that there are {\em at most} $2\ell$ such locations), the $4 \ell$ even point locations, and as many as dummy locations that it can still include.  The cost of the optimal solution will be precisely the cost of linking the points in the rest of the dummy locations to their nearest centers (at the even locations), which can be written as
$$\OPT = 1/2 \cdot (k/3 - (k/3 - F_0)) = 1/2 \cdot F_0 \le \ell,$$
where $F_0$ is the number of locations in $\{p^1, q^1, \ldots, p^\ell, q^\ell\}$ that have weights at least $1$.

Now suppose our solution ($\SOL$) outputs $c k$ centers for a constant $c \in [1, 1.01]$. Each time we include a  location $q^j$ as a center when there is no $0$-coordinate in the input column $X^j$, we have a loss of $1/2$ since we miss out on including a dummy location (i.e., we can take one more dummy location instead of taking $q^j$ as a center). Similarly, each time we do not include a location $q^j$ as a center when there is a $0$-coordinate in $X^j$, we have a loss of $1/2$ since a point at $q^j$ has weight at least $1$ but a point at a dummy location has weight at $1/2$. Therefore, even if we are allowed to output $c k$ medians, we will still need to figure out whether there is any point at location $q^j$ for at least an $\alpha = 0.9$ fraction of the coordinates $j \in [\ell]$. If not, then
\begin{eqnarray}
\SOL - \OPT &\ge& {1}/{2} \cdot (1 - \alpha) \ell - 1 \cdot (c - 1) k \label{eq:a-1} \ \ \ \ \\
&= & \frac{(1 - \alpha) - 12(c - 1)}{2} \cdot \ell \nonumber \\
&\ge& (6.1 - 6c) \OPT, \nonumber
\end{eqnarray}
where the first term in the RHS of (\ref{eq:a-1}) counts the loss of incorrectly computing the (at least) $(1-\alpha) \ell$ coordinates $j \in [\ell]$, and the second term counts the maximum gain of the extra $(c - 1) k$ centers $\SOL$ can use (compared with $\OPT$). 

By Lemma~\ref{lem:direct-sum}, we have that for any $c \in [1, 1.01]$, computing $(7.1 - 6c, c)$-bicriteria-approximation for $k$-median or $k$-means in the message passing model correctly with probability $0.9$ under distribution $X \sim \mu$ needs $\Omega(sk)$ bits of communication.  The theorem follows by Yao's minimax principle.
\end{proof}

\subsection{The blackboard model}
\label{sec:geo-bb}
%In Appendix~\ref{app:alg-geometric} 
We can show that there is an algorithm that achieves an $O(1)$-approximation using $\tilde{O}(s + k)$ bits of communication for $k$-median and $k$-means.  %Due to space constraints we defer the description of the algorithm to the full version of this paper.
For $k$-center, it is straightforward to implement the parallel guessing algorithm in the blackboard model using $\tilde{O}(s + k)$ bits of communication.

%\begin{thm}
%\label{thm:geometric-bb}
%There are algorithms that compute $O(1)$-approximations for $k$-median, $k$-means and $k$-center correctly with probability $0.9$ in the blackboard model using $\tilde{O}(s+k)$ bits of communication.
%\end{thm}
%\he{we should write the following discussion as a proof. Currently there is no proof of Theorem 4.4. We should also add some discussion on the lower bound in the blackboard model.}

%\he{The following is from the appendix of our NIPS submission. I did not make any changes.}

Our algorithm for $k$-median/means is an easy adaptation of the {\em successive sampling} algorithm proposed by Mettu and Plaxton~\cite{MP04} in the (centralized) RAM model.  We first summarize their algorithm and then describe how to port it to the blackboard model.

Let $X_1, \ldots X_s$ be the point sets at sites $\mathcal{P}_1, \ldots, \mathcal{P}_k$ respectively. The successive sampling algorithm proceeds in rounds. At each round $j$ it does the following:
\begin{enumerate}
\item $s$ sites jointly sample $O(k)$ point centers, denoted by $Y_j$;

\item $s$ sites grow balls from each of the point centers in $Y_j$ synchronously until a time step when a $0.9$ fraction of points in $\bigcup_{i \in [s]} X_i$ are covered;

\item each site $\mathcal{P}_i$ updates $X_j$ by removing those points that are covered by any of the balls centered at points in $Y_j$;

\item $s$ sites remove all the points covered by balls centered at points in $Y_j$, and proceed to the next round $j+1$.
\end{enumerate}
It is easy to see that the computation will finish in $r = O(\log n)$ rounds since at each round we remove a constant fraction of points.  At the end we compute an $O(1)$-approximation of $k$-median or $k$-means on the $O(k \log n)$ points $\bigcup_{j \in [r]} Y_j$.  In \cite{MP04} it has been shown that this algorithm gives an $O(1)$-approximation to $k$-median or $k$-means with high probability.

We now describe how to implement this centralized algorithm in the blackboard model.  We first consider each round.  Step 1 can be done by the distributed sampling algorithm in \cite{CMYZ12} using $\tilde{O}(k+s)$ bits of communication; note that at the end of this step the sampled points in $Y_j$ are written on the blackboard.
Step $2$ can be done by a binary search for the minimum ball radius $t_j$ such that $\bigcup_{p \in Y_j} \mathtt{Ball}(p, t_j)$ covers at least a $0.9$ fraction of points in  $\bigcup_{i \in [s]} X_i$, where $\mathtt{Ball}(p, t_j)$ denotes the ball centered at $p$ with radius $t_j$;  this binary search can be done using $\tilde{O}(1)$ bits of communication.  Step $3$ and $4$ can be done locally without any communication.   After $r$ rounds, the final clustering step can be done by any of the $s$ sites since all points in $\bigcup_{j \in [r]} Y_j$ have already been written on the blackboard.
Therefore the total communication cost can be bounded by $\tilde{O}(k+s)$.

Finally, we would like to mention that $\Omega(k + s)$ is an obvious lower bound, and thus our upper bound is tight up to logarithmic factors.  To see this, notice that $k$ is the size of the output, and the coordinator has to communication with each of the $s$ sites for at least $1$ bit.

%% file: experiment.tex
\newcommand{\twomoons}{{\tt Twomoons}}
\newcommand{\gauss}{{\tt Gauss}}
\newcommand{\sculpture}{{\tt Sculpture}}
\newcommand{\baseline}{{\tt Baseline}}
\newcommand{\MM}{{\tt MsgPassing}}
\newcommand{\blackboard}{{\tt Blackboard}}
\newcommand{\ncut}{\text{ncut}}
\newcommand{\chensays}[2][]{\textcolor{blue} {\textsc{Jiecao #1:} \emph{#2}}}

\section{Experiments}
In this section we present experimental results for  graph clustering in the message passing and blackboard models. We will compare the following three algorithms. (1) \baseline: each site sends all the data to the coordinator directly; (2) \MM: our algorithm in the message passing model (Section~\ref{sec:gcmessage}); (3) 
\blackboard: our algorithm in  the blackboard model (Section~\ref{sec:bb}).

%Since both of our algorithms are crucially based on the use of spectral scarification, our main focus in the experiments is to investigate to what extend the quality of the spectral clustering algorithms will be affected by using spectral sparsification, the saving of communication costs by using spectral sparsificaion, ...
%
%
%The goal of this experiment is not to demonstrate the effectiveness of the spectral clustering algorithm. We mainly want to investigate the following, 
%\begin{itemize}
%\item to what extend the quality of clustered results will be affected by using spectral sparsification.
%\item saving of communication costs by using spectral sparsifier.
%\item the affect of constants in algorithms of the message passing/blackboard model.
%\end{itemize}
%
%
%\subsection{The Setup}
%\paragraph{Reference Algorithms}
%We compare different algorithms in our experiment.

%Note that we can also run \MM~ in the blackboard model.

Besides giving the visualized results of these algorithms on various datasets, we also measure the qualities of the results via the {\em normalized cut}, defined as 
\[
\ncut(A_1, \ldots, A_{k}) = \frac{1}{2}\sum_{i\in[k]}\frac{w(A_i, V\backslash A_i)}{\vol(A_i)},
\]
 which is a standard objective function to be minimized for spectral clustering algorithms. 
%We will compare the communication costs of these algorithms in different settings.

%We also compare the total communication costs of different algorithms/models. As the unit does not matter in our case, we normalize all communication costs by the cost of \baseline.  Whenever possible, we will visualize the clustered results.

We implemented the algorithms using multiple languages, including Matlab, Python and C++. Our experiments were conducted on an IBM NeXtScale nx360 M4 server, which is equipped with 2 Intel Xeon E5-2652 v2 8-core processors, 32GB RAM and 250GB local storage.

\subsection{Datasets.}
We test the algorithms in the following real and synthetic datasets, which is visualized in \figref{visualization}.

\begin{figure}[h]
     \centering
     \subfigure[\twomoons]{\includegraphics[width=0.23\textwidth]{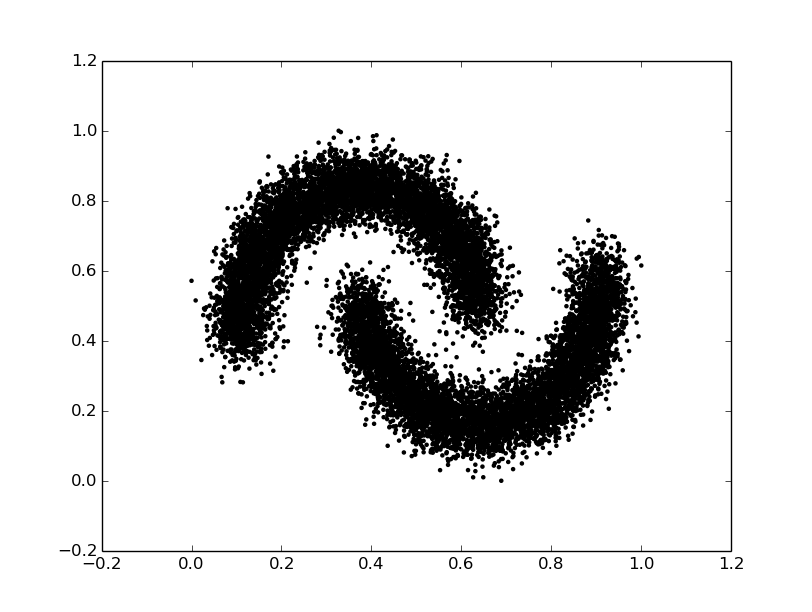}\label{fig:twomoons}}
     ~~
     \subfigure[\gauss]{\includegraphics[width=0.23\textwidth]{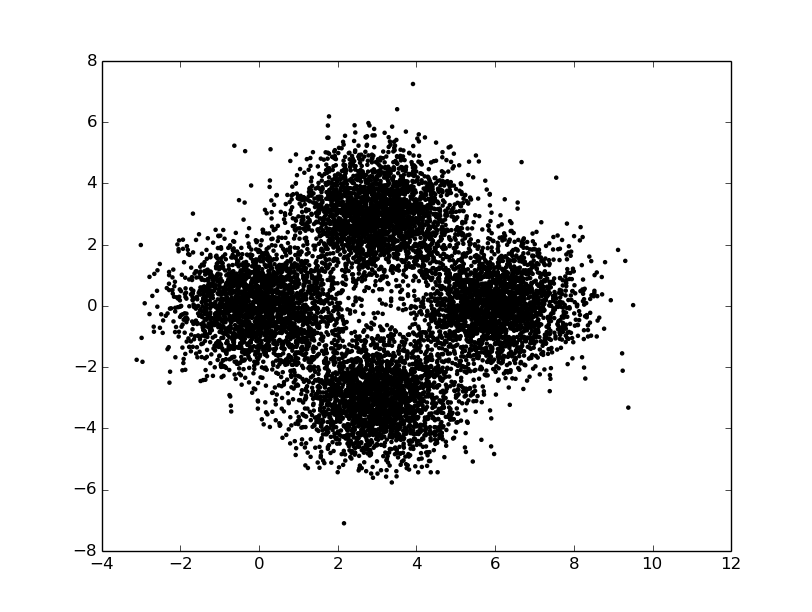}\label{fig:gauss}}
     ~~
     \subfigure[\sculpture]{\includegraphics[width=0.13\textwidth,height=0.16\textwidth]{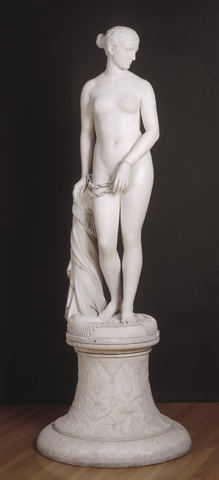}\label{fig:sculpture}}
     \caption{Visualization of the datasets for our experiments.}
     \label{fig:visualization}
\end{figure}

\vspace{-1mm}
\begin{itemize}
\item \twomoons : this dataset contains $n=14,000$ coordinates in $\mathbb{R}^2$. We consider each point to be a vertex. For any two vertices $u, v$, we add an edge with weight $w(u,v) = \exp\{-\|u-v\|_2^2/\sigma^2\}$ with $\sigma = 0.1$ when one vertex is among the $7000$-nearest points of the other.  This construction results in a graph with about $110,000,000$ edges.

\item  \gauss : this dataset contains $n = 10,000$ points in $\mathbb{R}^2$. There are $4$ clusters in this dataset, each generated using a Gaussian distribution. We construct a complete graph as the similarity graph.  For any two vertices $u, v$, we define the weight $w(u,v) = \exp\{-\|u-v\|_2^2/\sigma^2\}$ with $\sigma = 1$. The resulting graph has about $100,000,000$ edges.

\item \sculpture : a photo of \textit{The Greek Slave}~\footnote{Available in e.g., \url{http://artgallery.yale.edu/collections/objects/14794}}. We use an $80\times 150$ version of this photo where each pixel is viewed as a vertex. To construct a similarity graph, we map each pixel to a point in $\mathbb{R}^5$, i.e., $(x, y, r, g, b)$, where the latter three coordinates are the RGB values. For any two vertices $u, v$, we  put an edge between $u, v$ with weight $w(u,v) = \exp\{-\|u-v\|_2^2/\sigma^2\}$ with $\sigma = 0.5$ if one of $u, v$ is among the $5000$-nearest points of the other. This results in a graph with about $70,000,000$ edges.
\end{itemize}
\vspace{-1mm}
In the distributed model edges are randomly partitioned across $s$ sites. 

%\vspace{-1.5mm}

\subsection{Results on clustering quality}
%{\em Quality.} \
\begin{figure*}[ht]
     \centering
     \subfigure[\baseline]{\includegraphics[width=0.2\textwidth]{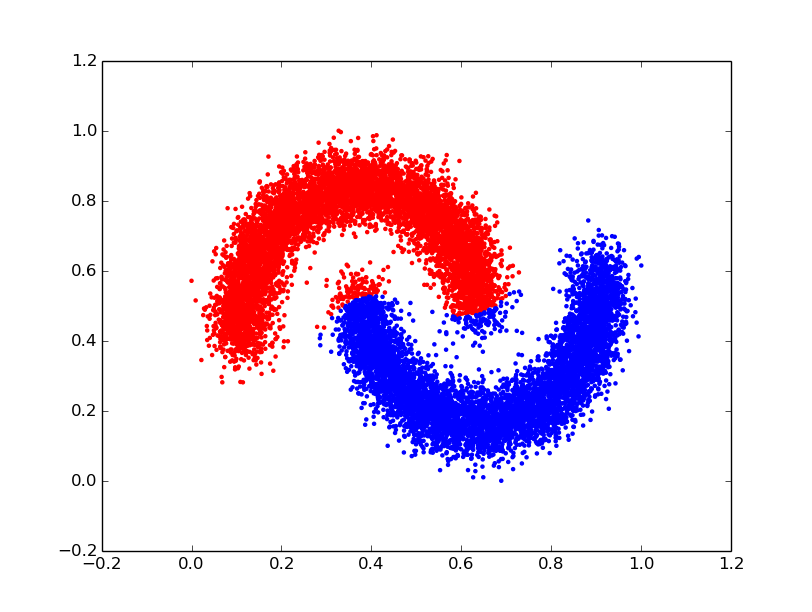}\label{fig:twomoons-clustered-original}}
     \subfigure[\MM]{\includegraphics[width=0.2\textwidth]{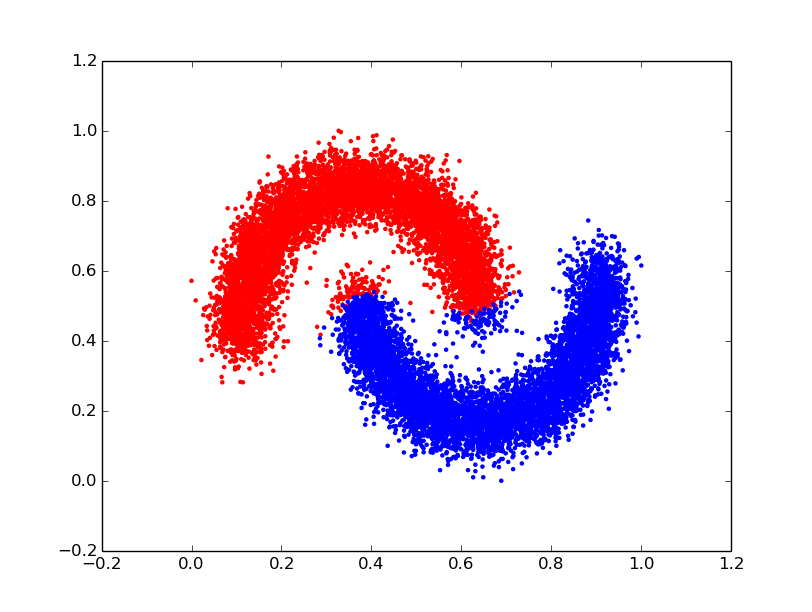}\label{fig:twomoons-clustered-sparsify}}
     \subfigure[\blackboard]{\includegraphics[width=0.2\textwidth]{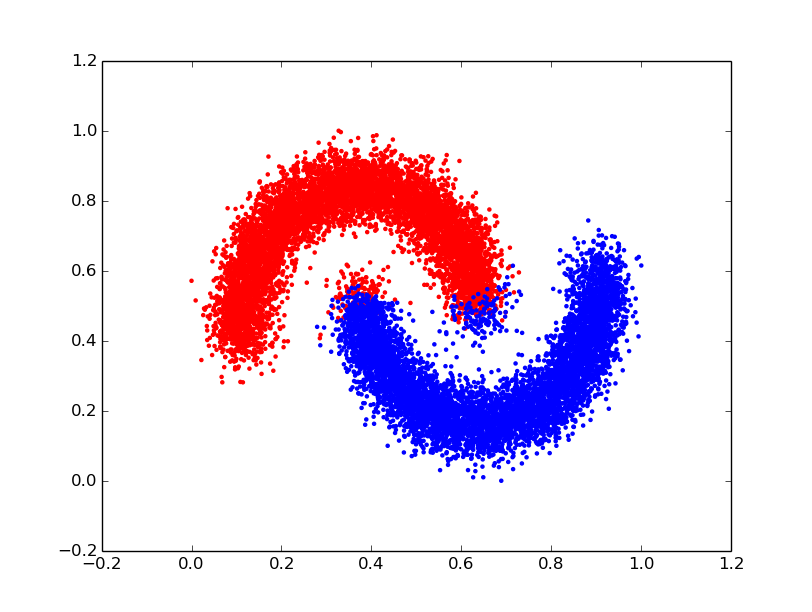}\label{fig:twomoons-clustered-chain}}
     \caption*{\twomoons, $k = 2$;}

\subfigure[\baseline]{\includegraphics[width=0.2\textwidth]{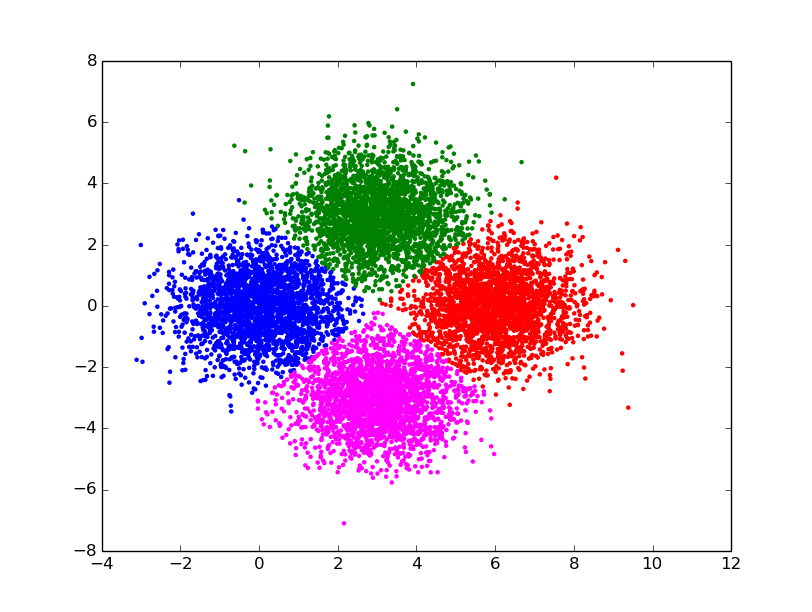}\label{fig:gauss-clustered-original}}
     \subfigure[\MM]{\includegraphics[width=0.2\textwidth]{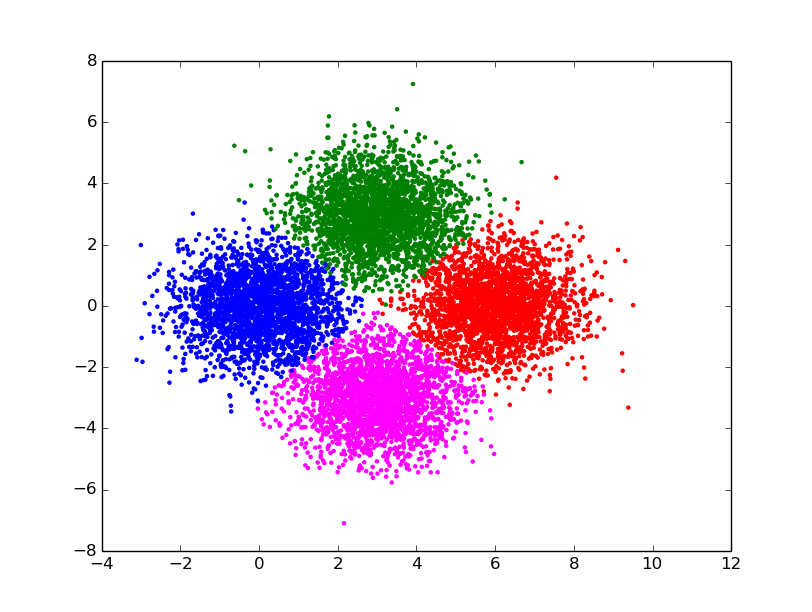}\label{fig:gauss-clustered-sparsify}}
     \subfigure[\blackboard]{\includegraphics[width=0.2\textwidth]{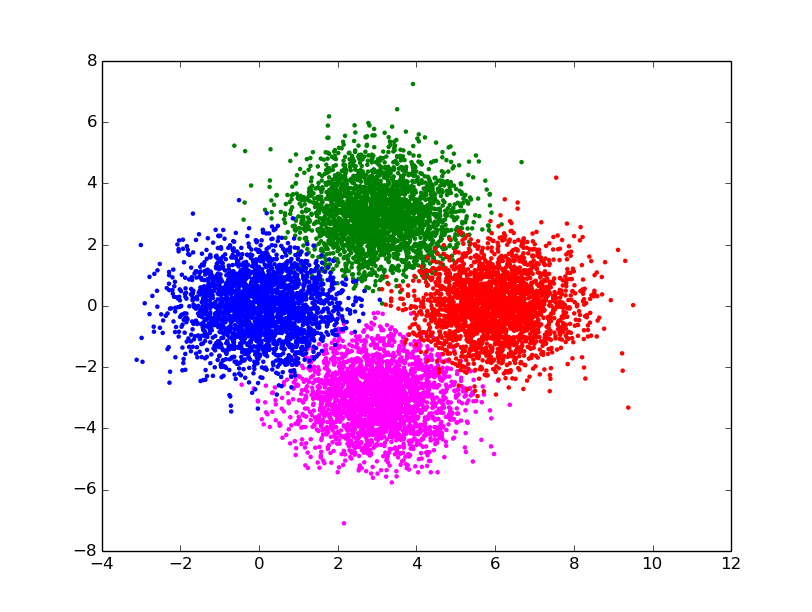}\label{fig:gauss-clustered-chain}}
     \caption*{\gauss, $k = 4$}

     \subfigure[\baseline]{\includegraphics[width=0.2\textwidth,height=0.2\textwidth]{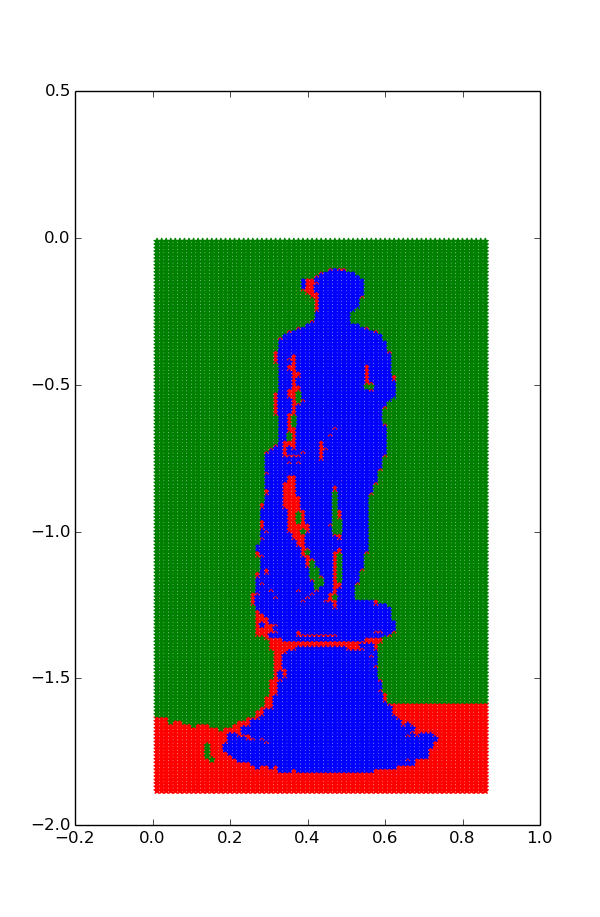}\label{fig:sculpture-clustered-original}}  
     \subfigure[\MM]{\includegraphics[width=0.2\textwidth,height=0.2\textwidth]{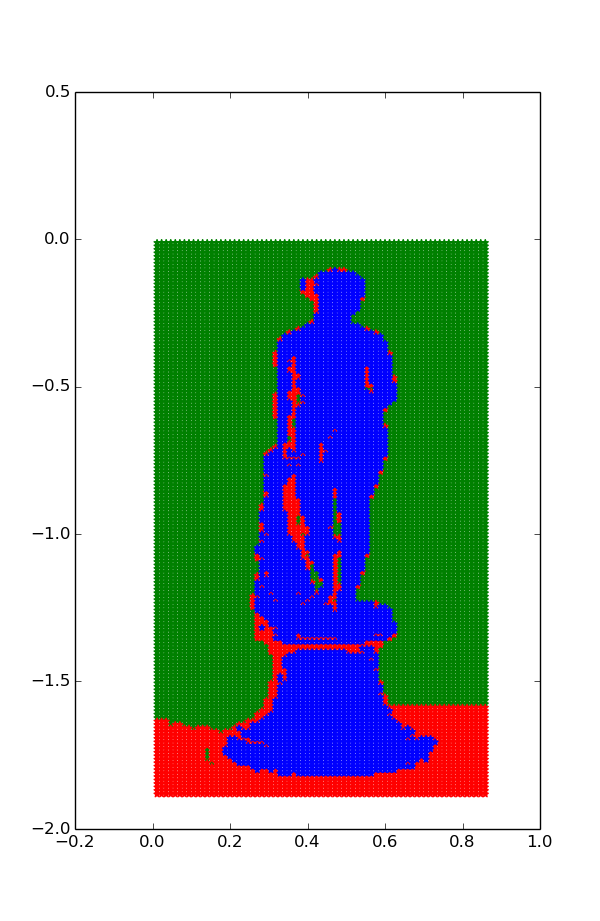}\label{fig:sculpture-clustered-sparsify}}
     \subfigure[\blackboard]{\includegraphics[width=0.2\textwidth,height=0.2\textwidth]{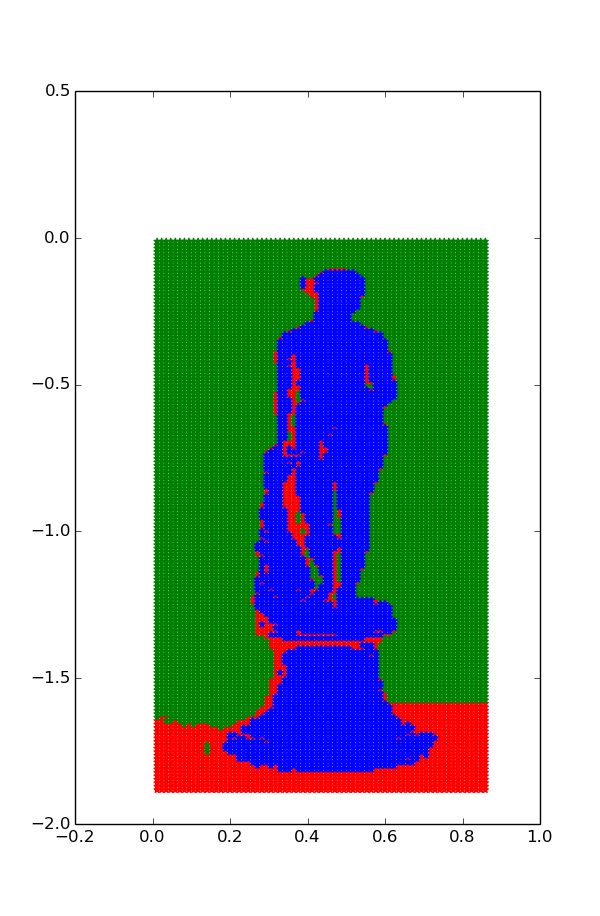}\label{fig:sculpture-clustered-chain}}
     \caption*{\sculpture, $k = 3$. }

     \caption{Visualization of the results on \twomoons, \gauss\ and \sculpture. In the message passing model each site samples $5 n$ edges; in the blackboard model all sites jointly sample $10n$ edges (in \twomoons~ and \gauss) or $20n$ edges (in \sculpture) and the chain has length $18$. $s = 15$.}
     \label{fig:quality-1}
\end{figure*}

We visualize the clustered results for 
the \twomoons, \gauss\ and \sculpture\ in Figure~\ref{fig:quality-1}.
% and visualize the clustered results for \gauss\ and \sculpture in Figure~\ref{fig:quality-2}.
It can be seen that \baseline, \MM\ and \blackboard\ give results of very similar qualities.  For simplicity, here we only present the visualization for $s=15$. Similar results were observed when we varied the values of $s$.  
%\he{To Qin: Do you plan to have two titles (Results \& Quality)?}

% \begin{figure*}[h]
%      \centering
% \subfigure[\baseline]{\includegraphics[width=0.3\textwidth]{gauss-10000-original-clustered.png}\label{fig:gauss-clustered-original}}
%      \subfigure[\MM]{\includegraphics[width=0.3\textwidth]{gauss-10000-sparsify-clustered-15.png}\label{fig:gauss-clustered-sparsify}}
%      \subfigure[\blackboard]{\includegraphics[width=0.3\textwidth]{gauss-10000-chain-clustered.png}\label{fig:gauss-clustered-chain}}
%      \caption*{\gauss, $k = 4$}

%      \subfigure[\baseline]{\includegraphics[width=0.2\textwidth]{sculpture-11680-original-clustered.png}\label{fig:sculpture-clustered-original}}  
%      \subfigure[\MM]{\includegraphics[width=0.2\textwidth]{sculpture-11680-sparsify-clustered-15.png}\label{fig:sculpture-clustered-sparsify}}
%      \subfigure[\blackboard]{\includegraphics[width=0.2\textwidth]{sculpture-11680-chain-clustered.png}\label{fig:sculpture-clustered-chain}}
%      \caption*{\sculpture, $k = 3$. }

%      \caption{Visualization of results on \gauss\ and \sculpture; in the message passing model each site samples $5 n$ edges; in the blackboard model all sites jointly sample $10n$ (in \gauss) or $20n$ (in \sculpture) edges and the chain has length $18$.}
%      \label{fig:quality-2}
% \end{figure*}

We also compare the normalized cut (ncut) values of the clustering results of different algorithms.  The results are presented in Figure \ref{fig:quality}. In all datasets, the ncut values of different algorithms are very close. The ncut value of \MM\ slightly decreases when we increase the value of $s$, while the ncut value of \blackboard\ is independent of $s$.
%We comment that in general, it is difficult to compare \MM\ and \blackboard\ directly because they are affected by different parameters.

\begin{figure*}[!ht]
  \centering
  \subfigure[\twomoons]{\includegraphics[width=0.33\textwidth]{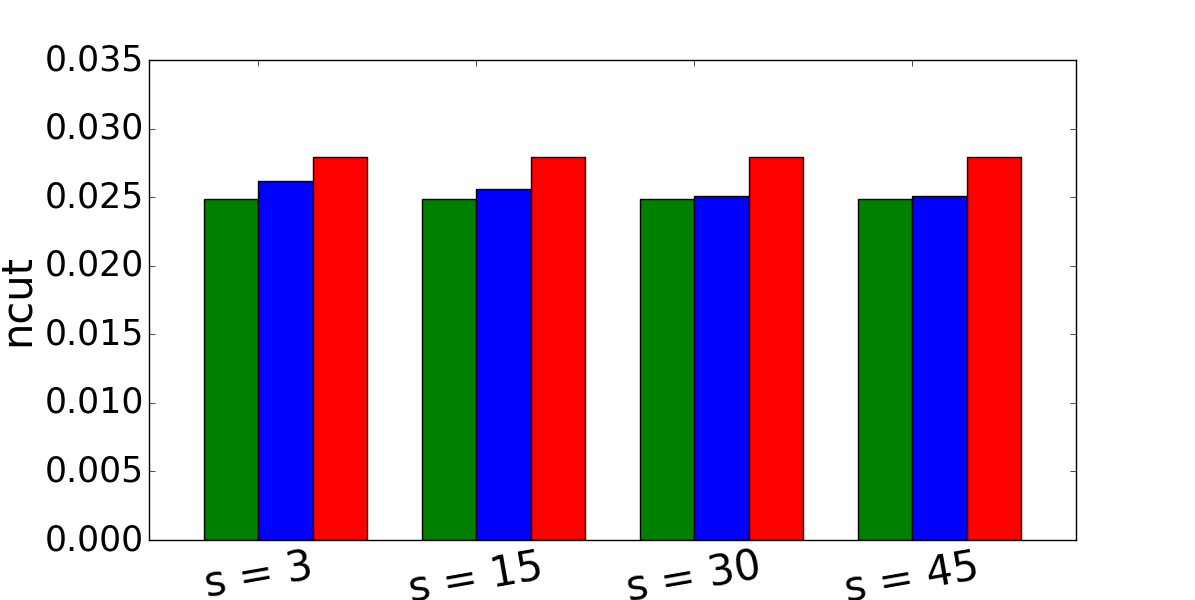}\label{fig:twomoons-quality}}\hspace*{-1.1em}
  \subfigure[\gauss]{\includegraphics[width=0.31\textwidth]{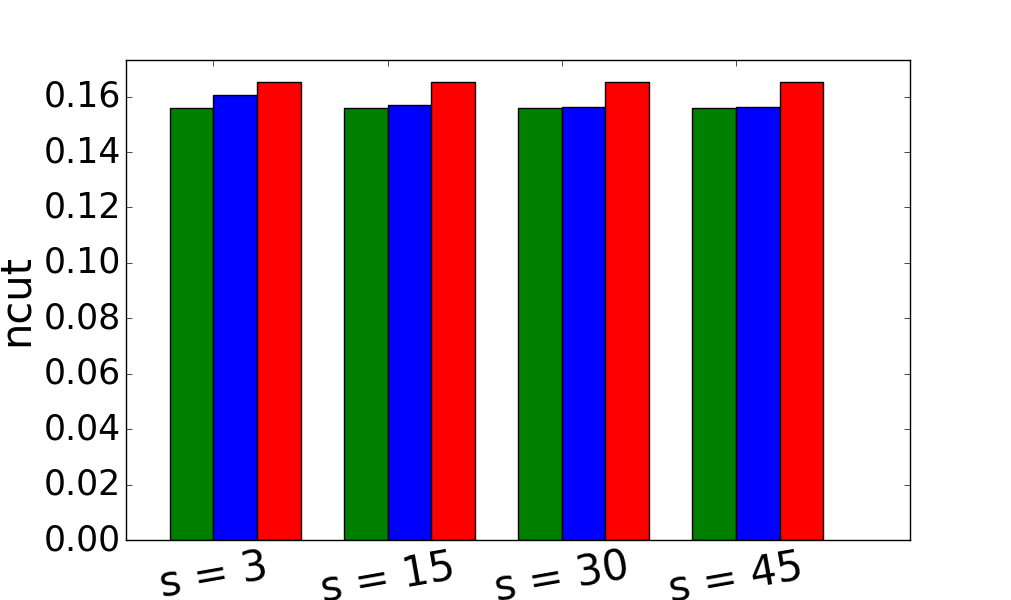}\label{fig:gauss-quality}}\hspace*{-1.1em}
  \subfigure[\sculpture]{\includegraphics[width=0.31\textwidth]{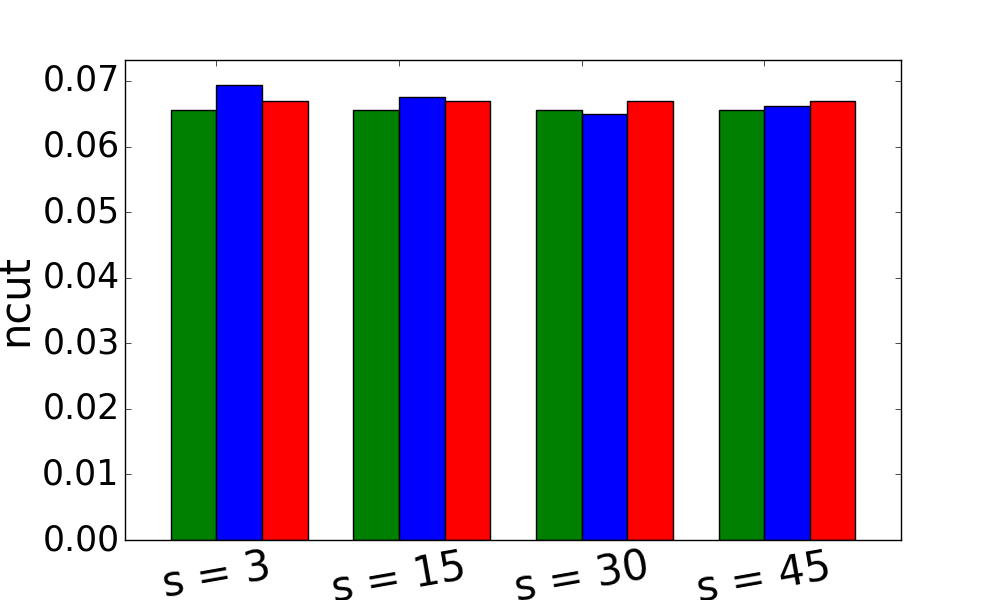}\label{fig:sculpture-quality}}\hspace*{-1.1em}
  \subfigure{\includegraphics[width=0.14\textwidth]{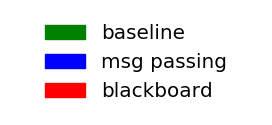}}
     \caption{Comparisons on normalized cuts. In the message passing model, each site samples $5n$ edges; in each round of the algorithm in the blackboard model, all sites jointly sample $10n$ edges (in \twomoons~and \gauss) or $20n$ edges (in \sculpture) edges and the chain has length $18$.}
     \label{fig:quality}
\end{figure*}

%\textcolor{red}{To Jiecao: Can you put the color lines indicating baseline, message passing, and blackboard within one row in Pic 2? Withthis we can save some space.}

%\vspace{-1.5mm}

\subsection{Results on communication costs} 
\begin{figure*}[!ht]
     \centering
     \subfigure[\twomoons]{\includegraphics[width=0.3\textwidth]{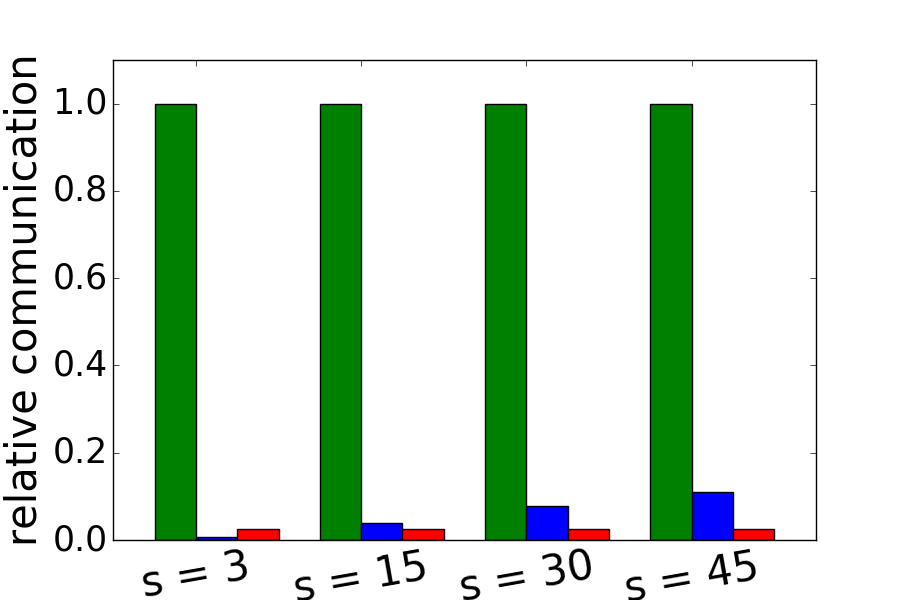}\label{fig:twomoons-communication}}
     \subfigure[\gauss]{\includegraphics[width=0.3\textwidth]{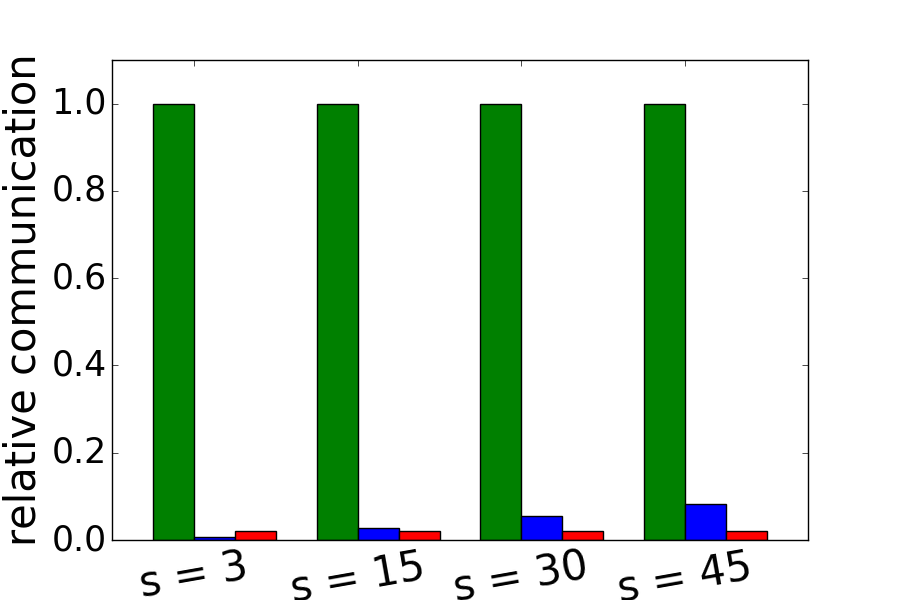}\label{fig:gauss-communication}}
     \subfigure[\sculpture]{\includegraphics[width=0.3\textwidth]{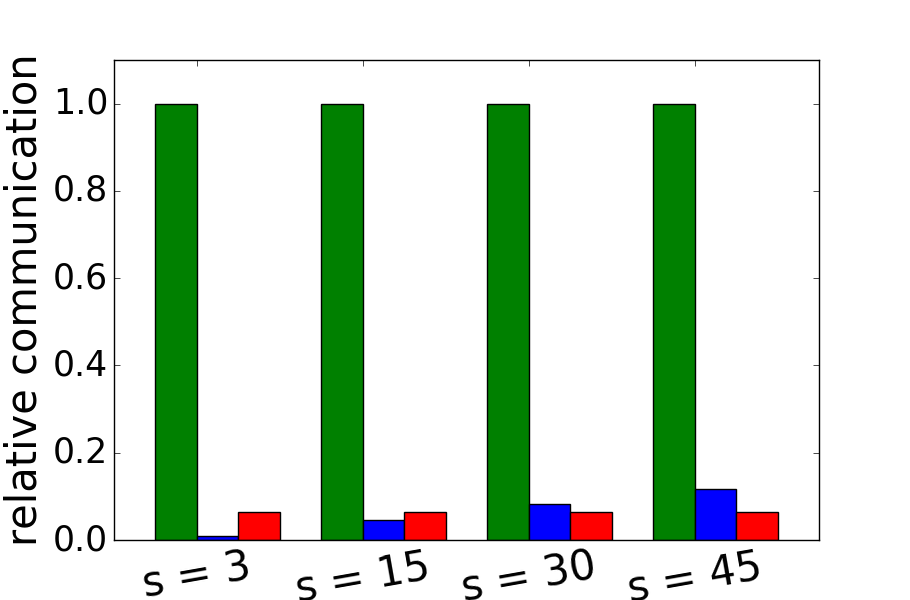}\label{fig:sculpture-communication}}

     \subfigure[\twomoons]{\includegraphics[width=0.32\textwidth]{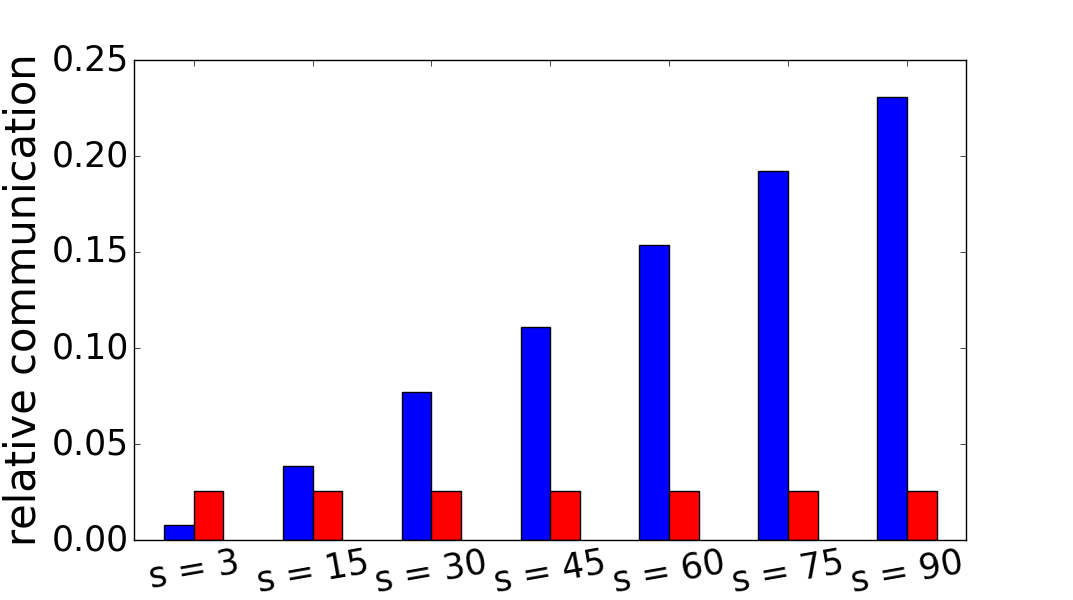}\label{fig:twomoons-communication-2}}
     \subfigure[\gauss]{\includegraphics[width=0.32\textwidth]{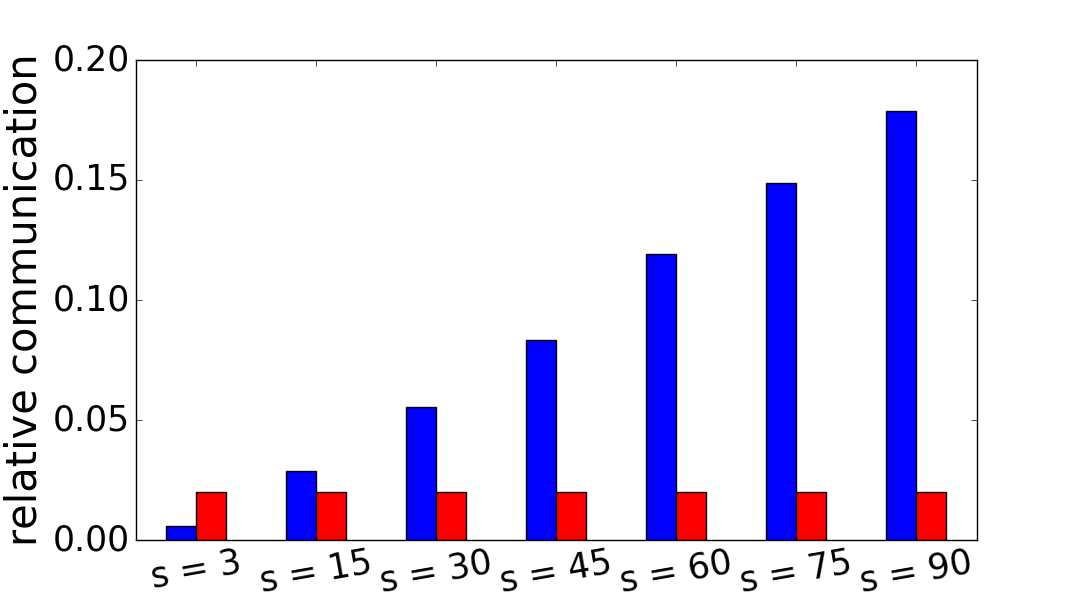}\label{fig:gauss-communication-2}}
     \subfigure[\sculpture]{\includegraphics[width=0.32\textwidth]{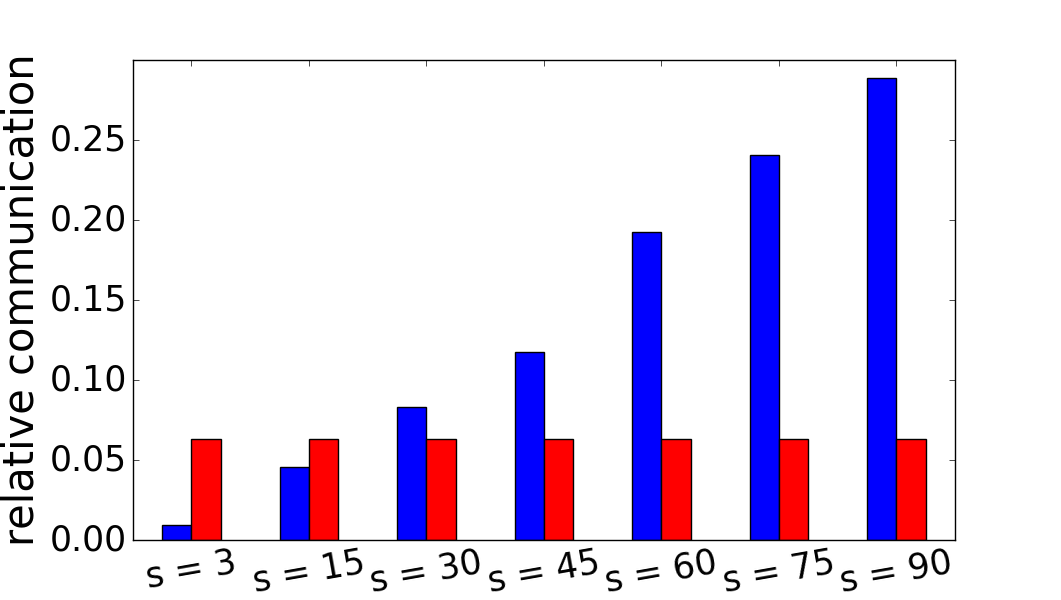}\label{fig:sculpture-communication-2}}
     \caption{Comparisons on communication costs. In the message passing model, each site samples $5n$ edges; in each round of the algorithm in the blackboard model, all sites jointly sample $10n$ (in \twomoons~and \gauss) or $20n$ (in \sculpture) edges and the chain has length $18$. }
     \label{fig:communication}
\end{figure*}

We compare the communication costs of different algorithms in Figure \ref{fig:communication}. We observe that while achieving similar clustering qualities as \baseline, both \MM\ and \blackboard\ are significantly more communication-efficient (by one or two orders of magnitudes in our experiments). We also notice that the value of $s$ does not affect the communication cost of \blackboard, while the communication cost of \MM\ grows almost linearly with $s$; when $s$ is large, \MM\ uses significantly more communication than \blackboard. These confirm our theory.  %In Figure~\ref{fig:mm-const} and Figure~\ref{fig:blackboard-const}   in Appendix~\ref{sec:parameters} we present how the performance of \MM\ and \blackboard\ are affected by their parameters.

%
%
%\vspace{-1.5mm}
%\paragraph{Summary.}  From our experimental results we conclude that \MM\ and \blackboard\ achieve similar clustering quality as the native algorithm \baseline, while significantly reduce the communication cost.  When the number of sites is large, \blackboard\ is more communication efficient than \MM, as predicted by our theory.

\subsection{Parameters in \MM\ and \blackboard}
\label{sec:parameters}

Figure \ref{fig:mm-const} shows in \MM how the value of ncut is affected by the number of sites and the number of edges sampled in each site. 
Here, each site samples $cn$ edges. 
When $c=3$ and $s=1$, the ncut value diverges in all datasets. This is because with such a small $c$, the algorithm does not generate a valid sparsifier. In general, increasing $c$ or $s$ will slightly decrease the ncut value. But once they are above some thresholds, the ncut values of \MM\ and \baseline\ become very close.

Figure \ref{fig:blackboard-const} shows in \blackboard  how the ncut value is affected by the number of iterations and the number of edges sampled. When the number of iterations is set to be $5$, ncut values diverge in all datasets. This is because we cannot expect to generate a valid sparsifier by using such few iterations. It can be seen from \ref{fig:bb-gauss-constant} that for a fixed $c$, performing more iterations will help to reduce ncut values. From the same figure, one can also conclude that for fixed iterations, increasing $c$ also helps to reduce the ncut values.

\begin{figure*}[h!t]
     \centering
     \subfigure[\twomoons]{\includegraphics[width=0.3\textwidth]{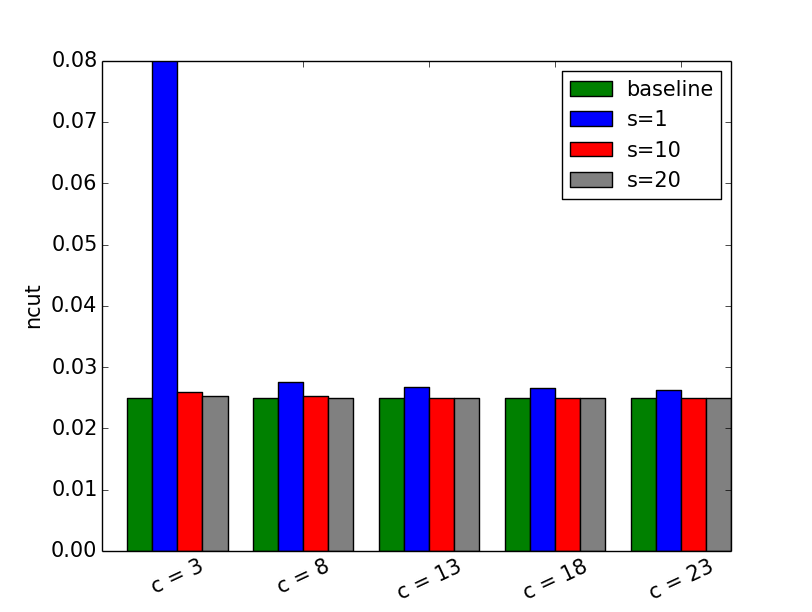}\label{fig:mm-twomoons-constant}}
     \subfigure[\gauss~dataset]{\includegraphics[width=0.3\textwidth]{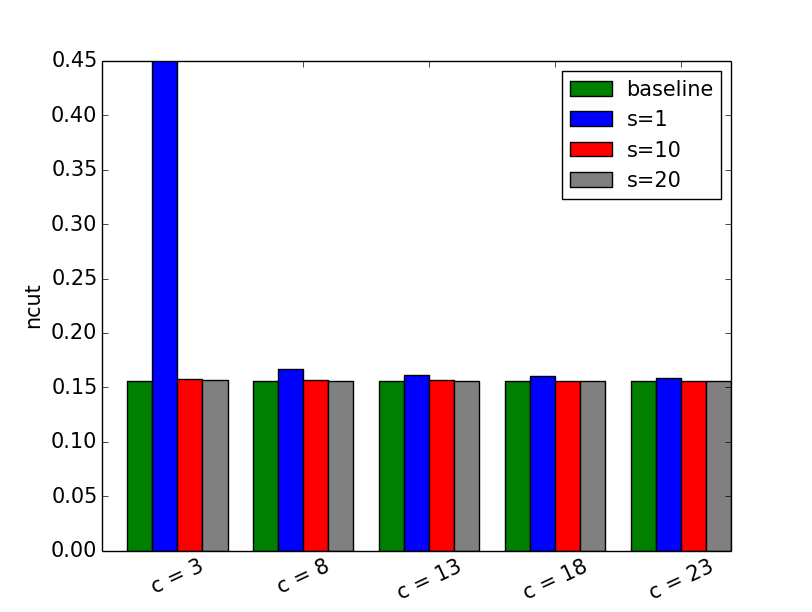}\label{fig:mm-gauss-constant}}
     \subfigure[\sculpture]{\includegraphics[width=0.3\textwidth]{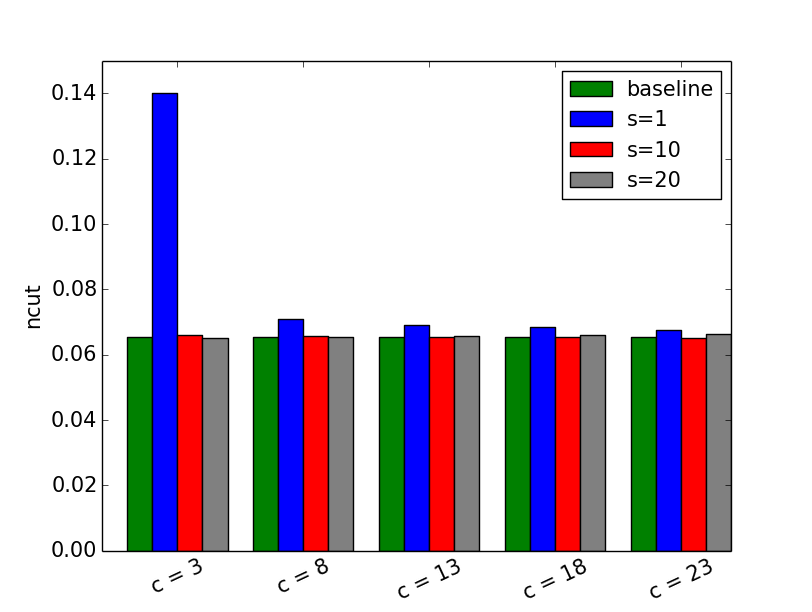}\label{fig:mm-sculpture-constant}}
     \caption{The pictures above show the $\ncut$ values with respect to the values of $c$ and $s$ for the \MM\ algorithm. Here  
 each site samples $c n$ edges.}
     \label{fig:mm-const}
\end{figure*}

\begin{figure*}[h!t]
     \centering
     \subfigure[\twomoons]{\includegraphics[width=0.3\textwidth]{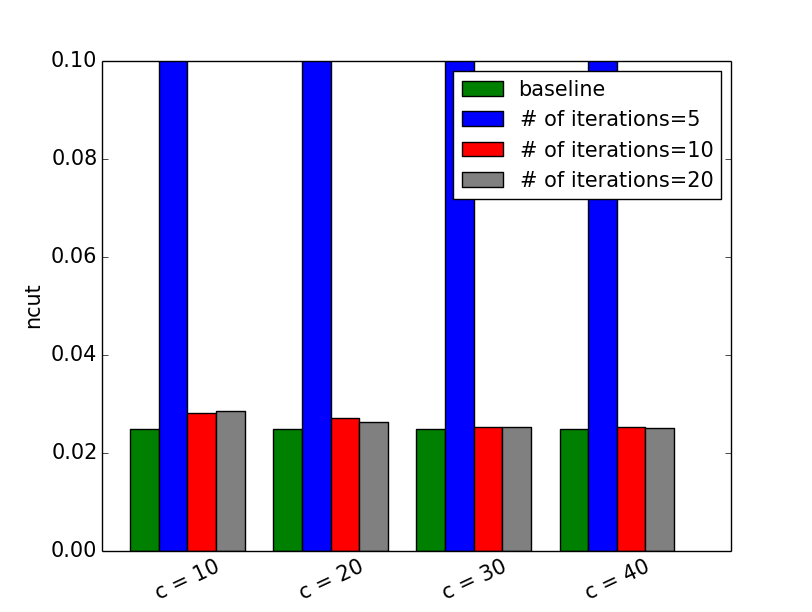}\label{fig:bb-twomoons-constant}}
     \subfigure[\gauss]{\includegraphics[width=0.3\textwidth]{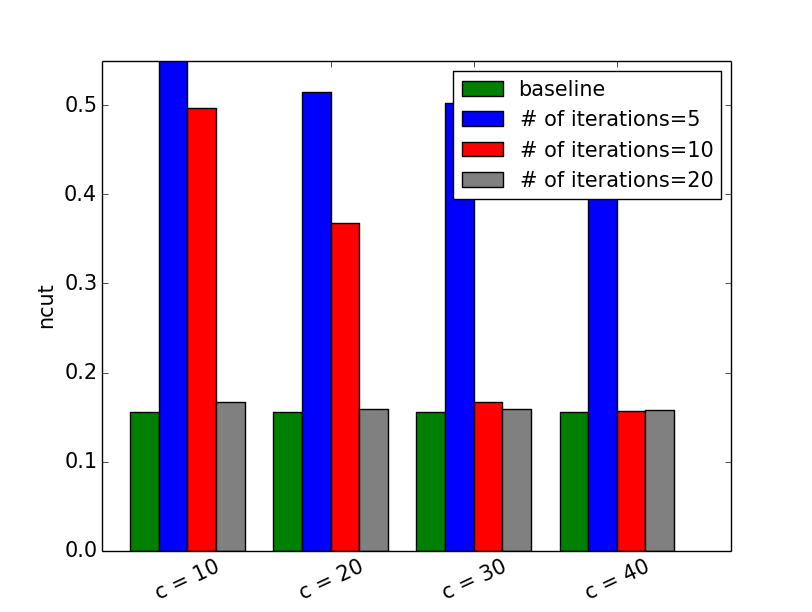}\label{fig:bb-gauss-constant}}
     \subfigure[\sculpture]{\includegraphics[width=0.3\textwidth]{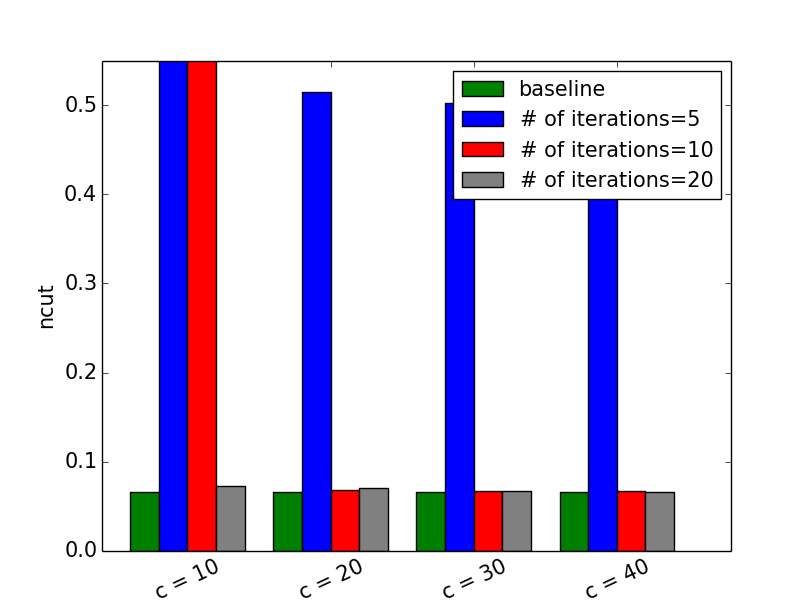}\label{fig:bb-sculpture-constant}}
     \caption{The pictures above show how the $\ncut$ values are affected by the number of iterations and the value of $c$ for the \blackboard\ algorithm. Here 
all sites jointly sample $c n$ edges. }
     \label{fig:blackboard-const}
\end{figure*}